\documentclass[sigconf, nonacm]{acmart}

\usepackage{algorithmic}
\usepackage[ruled,vlined]{algorithm2e}
\usepackage{booktabs}
\usepackage{pifont}
\usepackage{fancyhdr}
\usepackage{color}
\usepackage{soul, framed}
\usepackage{makecell}
\usepackage{xcolor}

\usepackage{colortbl}
\usepackage{graphicx}
\usepackage{textcomp}
\usepackage{xcolor,colortbl}
\usepackage{multirow}
\usepackage{bm,dsfont}
\usepackage{subfigure}
\usepackage{balance}
\usepackage{stmaryrd}
\usepackage{autobreak}
\usepackage{enumitem}
\usepackage{svg}
\usepackage{caption}
\usepackage{bbding}
\usepackage{wasysym}
\usepackage{textcomp}
\usepackage{fontawesome}
\usepackage{tcolorbox}

\definecolor{teal}{RGB}{0, 128, 128} 

\newcommand{\framework}{{EnMCS}}

\newcommand{\singlesearchname}{{HoloSearch}}
\newcommand{\combinename}{EMerge}

\newcommand{\annot}[1]{\quad\triangleright \text{#1}}

\newcommand\vldbavailabilityurl{URL_TO_YOUR_ARTIFACTS}

\newcommand\vldbpagestyle{plain} 

\begin{document}

\title{Ensemble-based Deep Multilayer Community Search}

\author{Jianwei Wang$^{1}$, ~~ Yuehai Wang$^{2}$,~~ Kai Wang$^{3}$, ~~ Xuemin Lin$^{3}$,~~Wenjie Zhang$^{1}$,~~Ying Zhang$^{4}$}
\affiliation{\normalsize{$^{1}${The University of New South Wales, Sydney, Australia}}\\
\normalsize{$^{2}${Shanghai Jiao Tong University, Shanghai, China}}\\
 \normalsize{$^{3}${Antai College of Economics \& Management, Shanghai Jiao Tong University, Shanghai, China}} \\
 \normalsize{$^{4}${Zhejiang Gongshang University, Hangzhou, China}}}
\email{{jianwei.wang1, wenjie.zhang}@unsw.edu.au, {w.kai, xuemin.lin}@sjtu.edu.cn, ying.zhang@zjgsu.edu.cn}

\renewcommand{\shortauthors}{Jianwei Wang et al.}
\renewcommand{\shorttitle}{Research Data Management Track Paper}

\begin{abstract}
Multilayer graphs, consisting of multiple interconnected layers, are widely used to model diverse relationships in the real world.
A community is a cohesive subgraph that offers valuable insights for analyzing (multilayer) graphs. Recently, there has been an emerging trend focused on searching query-driven communities within the multilayer graphs.
However, existing methods for multilayer community search are either 1) rule-based, which suffer from structure inflexibility; or 2) learning-based, which rely on labeled data or fail to capture layer-specific characteristics.
To address these, we propose \textbf{\framework}, an \underline{\textbf{En}}semble-based unsupervised (i.e., label-free) \underline{\textbf{M}}ultilayer \underline{\textbf{C}}ommunity \underline{\textbf{S}}earch framework. 
\framework~contains two key components, i.e., \singlesearchname~ which identifies potential communities in each layer while integrating both layer-shared and layer-specific information, and \combinename~which is an Expectation-Maximization (EM)-based method that synthesizes the potential communities from each layer into a consensus community.
Specifically, \singlesearchname~first employs a graph-diffusion-based model that integrates three label-free loss functions to learn layer-specific and layer-shared representations for each node. Communities in each layer are then identified based on nodes that exhibit high similarity in layer-shared representations while demonstrating low similarity in layer-specific representations w.r.t. the query nodes. 
To account for the varying layer-specific characteristics of each layer when merging communities, \combinename~ models the error rates of layers and true community as latent variables. It then employs the EM algorithm to simultaneously minimize the error rates of layers and predict the final consensus community through iterative maximum likelihood estimation.
Experiments over 10 real-world datasets highlight the superiority of \framework~ in terms of both efficiency and effectiveness.

\end{abstract}

 \maketitle
\pagestyle{\vldbpagestyle}

\ifdefempty{\vldbavailabilityurl}{}{
\begingroup\small\noindent\raggedright\textbf{PVLDB Artifact Availability:}\\
The source code, data, and/or other artifacts have been made available at \url{https://github.com/guaiyoui/EnMCS}.
\endgroup
}

\section{Introduction}
\label{sec:Introduction}

Graphs are widely used to model relationships among entities in various systems, such as social networks~\cite{wasserman1994social, knoke2008social, he2024discovering, he2023scaling, li2024querying, DBLP:conf/wise/LiYZZL21}, road networks~\cite{DBLP:journals/corr/abs-2408-05432,DBLP:conf/icde/WangYC00L23} and financial networks~\cite{eboli2004systemic, heimo2009maximal}. In many real-world scenarios involving complex networks, interactions between entities often span multiple domains (e.g., different types of sports as illustrated in Figure~\ref{fig:case_example}). Multilayer graph emerges as a powerful structure for modeling such diverse relationships~\cite{hashemi2022firmcore, behrouz2022firmtruss}, where each layer represents a different type of interaction. 

Applications such as friend recommendation in social networks~\cite{chen2008combinational, cheng2011personalized} and fraud detection on e-commerce platforms~\cite{li2021happens, zhang2017hidden} require the identification of query-driven cohesive subgraphs. To handle these demands, community search (CS)\cite{wang2024efficient, wang2024neural, fang2020survey, sozio2010community}, also known as local community detection, is proposed. Given a set of query nodes, CS aims to find a query-driven subgraph, where the resulting subgraph (or community) forms a densely intra-connected structure and contains the query nodes. Recently, there is a growing interest in developing CS on the multilayer graphs~\cite{hashemi2022firmcore, behrouz2022firmtruss, luo2020local, interdonato2017local}. Multilayer CS has broad applications in real-world scenarios, including:

\noindent $\bullet$ \textit{Identifying functional regions in brain networks}~\cite{luo2020local, behrouz2022firmtruss}. 
Multilayer brain graphs from multiple individuals enable more accurate and comprehensive modeling of functional systems, as subjects may share similar visual or auditory characteristics. Given a query region, multilayer CS can identify related regions with similar functionality across different subjects.

\noindent $\bullet$ \textit{Collaboration mining}~\cite{luo2020local}. In a multilayer collaboration graph, each node represents a researcher, edges indicate collaborations, and each layer corresponds to a research domain (e.g., database, machine learning). Given a specific researcher, multilayer CS can identify related researchers with similar expertise across layers.

\noindent $\bullet$ \textit{Social friendship mining}~\cite{kim2015community}. 
A multilayer graph can more effectively model individuals across various dimensions (e.g., location, interests, demographics). Multilayer CS in social multilayer graphs uncovers intricate community structures and reveals potential connections, enhancing the mining of social friendships.

\begin{figure}
  \centering
  \includegraphics[width=0.90\linewidth]{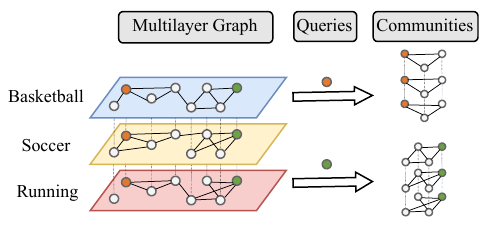}
  \vspace{-3mm}
  \caption{An example of multilayer community search. The left diagram shows a multilayer graph, while the right diagram shows the communities identified corresponding to queries indicated by the arrows.}
  \label{fig:case_example}
\vspace{-6mm}
\end{figure}

\noindent \textbf{Existing works}. Given the importance and broad applications of multilayer CS, various methods are proposed~\cite{hashemi2022firmcore, behrouz2022firmtruss, wang2024focuscore, interdonato2017local, luo2020local, behrouz2022cs}. 
Table~\ref{tab:intro_comparasion} provides an overview of representative approaches, which can be categorized into two main categories: rule-based methods~\cite{hashemi2022firmcore, behrouz2022firmtruss, wang2024focuscore, interdonato2017local} and learning-based methods~\cite{luo2020local, behrouz2022cs}.

Rule-based multilayer CS methods characterize communities by pre-defined patterns or rules. 
For instance, FirmCore~\cite{hashemi2022firmcore} and FocusCore~\cite{wang2024focuscore} adapt the concept of $k$-core~\cite{cui2014local} to guide multilayer CS. FirmCore requires that at least $\lambda$ of the layers in the multilayer graph satisfy the condition that nodes within the selected community have a degree greater than or equal to $k$. Similarly, FirmTruss~\cite{behrouz2022firmtruss} adopts the concept of $k$-truss~\cite{akbas2017truss} for multilayer CS, and ML-LCD~\cite{interdonato2017local} generalizes the cohesiveness function~\cite{andersen2006local} as the objective for multilayer CS. 
Although label-free, these methods suffer from structural inflexibility~\cite{behrouz2022cs} and fail to account for layer-specific characteristics. These fixed rules and patterns impose rigid constraints on the topological structure of communities and treat each layer identically, making it difficult for real-world communities to satisfy such inflexible constraints.

Different from rule-based methods that operate directly on the graph, learning-based methods project nodes into a latent space, where rules or objectives can be applied more flexibly and effectively. 
These methods typically first employ techniques such as random walks as in RWM~\cite{luo2020local} or multilayer graph convolutional networks (GCNs) as in CS-MLGCN~\cite{behrouz2022cs} to learn latent representations. Then, RWM uses a cohesiveness function (e.g., conductance) to identify communities, while CS-MLGCN uses ground-truth communities to train a classification model.
Although showing promise, the representations learned by RWM rely on the random walk and thus lack layer-specific characteristics and latent cohesiveness, which limits its effectiveness. Furthermore, the representations learned by CS-MLGCN rely on label information which are hard to obtain in real-world scenarios, restricting its applicability.

\noindent \textbf{Motivations}.
As summarized in Table~\ref{tab:intro_comparasion}, existing multilayer CS methods are either 1) rule-based, suffering from structural inflexibility; or 2) learning-based, relying on label information or lacking layer-specific characteristics. 
Recently, TransZero~\cite{wang2024efficient}, a learning-based single-layer CS method, is proposed.
It includes an offline pre-training phase and an online similarity-to-community phase, which frees the model from label reliance while achieving high performance.
However, TransZero is limited to single-layer graphs and cannot directly support multilayer graphs.
A direct approach to generalizing single-layer CS methods for multilayer graphs, as detailed in Algorithm~\ref{algo:naive_approach}, involves 1) searching for communities in each layer via the single-layer CS method, and 2) merging the communities into a consensus community. 
While the idea of "search and merge" has been widely applied in multilayer community detection~\cite{tagarelli2017ensemble, xu2020stacked}, the naive application of this strategy to multilayer CS lacks inter-layer complementary information in both the search stage and the merge stage, leading to unsatisfactory performance.
In this paper, we build upon and enhance this strategy, aiming to design an effective and efficient unsupervised multilayer CS method. Two main challenges exist below.

\begin{table}

 \caption{Comparisons among (multilayer) CS methods}
 \vspace{-3mm}
 \centering \scalebox{0.72}{
\begin{tabular}
{@{}lccccccc@{}}
\toprule[1.2pt]
Approaches  & \makecell[c]{Support \\ Multilayer?} & \makecell[c]{Layer-specific \\ Characteristics?}  &  \makecell[c]{Structure \\ Flexibility?}    & \makecell[c]{Latent \\ Cohesiveness?}    & \makecell[c]{Label \\ Free?} \\ \midrule
FirmCore~\cite{hashemi2022firmcore} & \Checkmark     & \XSolidBrush &   \XSolidBrush   & \XSolidBrush &\Checkmark        \\
FocusCore~\cite{wang2024focuscore} & \Checkmark     & \XSolidBrush &   \XSolidBrush   & \XSolidBrush &\Checkmark         \\
FirmTruss ~\cite{behrouz2022firmtruss}   & \Checkmark      & \XSolidBrush & \XSolidBrush       & \XSolidBrush  &   \Checkmark    \\
ML-LCD~\cite{interdonato2017local} &\Checkmark    & \XSolidBrush & \XSolidBrush & \XSolidBrush & \Checkmark    \\
RWM~\cite{luo2020local} & \Checkmark    & \XSolidBrush & \Checkmark & \XSolidBrush &  \Checkmark   \\
CS-MLGCN~\cite{behrouz2022cs}   &  \Checkmark  & \Checkmark & \Checkmark    & \Checkmark & \XSolidBrush     \\
TransZero~\cite{wang2024efficient}   & \XSolidBrush    & \XSolidBrush  & \Checkmark & \Checkmark & \Checkmark        \\
 
\midrule
\framework~(Ours)  & \Checkmark & \Checkmark & \makecell[c]{\Checkmark} & \Checkmark & \Checkmark\\ \bottomrule
\end{tabular}}

 \vspace{-3mm}
 \label{tab:intro_comparasion}
\end{table}
\textit{Challenge I: How to efficiently and effectively search communities in each layer while incorporating both inter-layer and intra-layer information}. 
The lack of ground-truth communities makes the search process challenging.
One direct approach is to leverage a global graph neural network (GNN), such as the multilayer GCNs in CS-MLGCN~\cite{behrouz2022cs}, to learn to predict community. However, training such an end-to-end model needs a specific optimization objective, e.g., the label used in CS-MLGCN.  
An alternative approach involves first learning representations by some unsupervised frameworks~\cite{zhu2022structure, wang2022collaborative}, followed by community search. However, the majority of existing representation learning frameworks for multilayer graphs prioritize global properties~\cite{zhu2022structure, wang2022collaborative}, overlooking layer-specific characteristics. Additionally, these frameworks usually rely on GCNs as their backbone, which primarily operate on local neighborhoods. This choice not only leads to low computational efficiency but also limits the receptive field. Therefore, it is challenging to search communities while incorporating multilayer information without labels.

\textit{Challenge II: How to effectively merge the communities across layers into a consensus community without labels}. The intricate inter-layer relationships and the lack of labels complicate the merging process. 
A direct approach is to use a majority voting strategy or weighted voting, where communities from each layer "vote" on node assignments to form a consensus community. 
However, a naive voting strategy lacks layer-specific characteristics, and without labels, determining accurate weights for voting is potentially unreliable.
Another promising direction is consensus clustering~\cite{tagarelli2017ensemble, hao2023ensemble,topchy2004mixture,wang2022stable} to ensemble communities. However, as detailed in Algorithm~\ref{algo:naive_approach} and in the experiments, consensus clustering, such as the co-association-based~\cite{tagarelli2017ensemble, hao2023ensemble} and the mixture-ensemble-based~\cite{topchy2004mixture, wang2022stable} methods, usually ensemble clusterings from multiple clustering algorithms on the same dataset and rely on the assumption of independence among clustering outcomes~\cite{topchy2004mixture}. In contrast, the communities from different layers have intricate inter-layer correlations. 
Thus, it is challenging to merge the communities without labels.

\noindent\textbf{Our solutions}. To mitigate the above challenges, we propose a new \underline{\textbf{En}}semble-based unsupervised \underline{\textbf{M}}ultilayer \underline{\textbf{C}}ommunity \underline{\textbf{S}}earch framework, termed \framework. \framework~comprises two key components: \singlesearchname, designed to handle Challenge I by proposing a graph diffusion encoder and inter-layer similarity-to-community module, and \combinename, aimed at solving Challenge II by incorporating layer-specific characteristics as latent variables and merging the communities using Expectation-Maximization (EM) algorithm.

In \singlesearchname, we first design a graph diffusion encoder to efficiently and effectively learn both layer-shared and layer-specific representations across layers. The learned representations are then used for inter-layer similar-to-community to identify communities. 
Specifically, the graph diffusion encoder employs a diffusion process~\cite{kondor2002diffusion, gasteiger2019diffusion} to aggregate information. By first applying diffusion process and then feedforward process, it efficiently and effectively leverages both global and local features to learn representations.
To train the graph diffusion encoder, we introduce three label-free loss functions following~\cite{mo2023disentangled, wang2024efficient}: inter-layer loss, intra-layer loss and proximity loss. The inter-layer loss ensures inter-layer consistency by aligning shared patterns across layers, and the intra-layer loss enhances discriminability by emphasizing layer-specific characteristics. The proximity loss promotes cohesiveness in the latent space to encourage closer representations of similar nodes. 
Then, the inter-layer similarity-to-community module identifies communities by selecting nodes with high similarity in layer-shared representations and low similarity in layer-specific representations w.r.t. the query.

In \combinename, rather than explicitly assigning a weight to each layer, we model the true community and layer-specific uncertainties, i.e., error rates of each layer as latent variables. Using the EM algorithm, we simultaneously estimate both the true community and the error rates through iterative optimization motivated by ~\cite{dawid1979maximum}. Specifically, \combinename~iteratively applies two main steps: the E-step and the M-step. In the E-step, given the current estimates of error rates and the category prior (i.e., probability of each node belonging to each category), we update the community memberships. In the M-step, we use maximum likelihood estimation to update the estimated error rate and category priors based on the estimated community memberships. This process is repeated iteratively until convergence, yielding a consensus community that integrates information across layers while accounting for layer-specific uncertainties, thereby achieving high accuracy. Furthermore, \combinename~does not rely on labeled data to merge communities, enhancing its applicability.

This design of independent search and merge allows for flexible updates of either \singlesearchname~or \combinename~to achieve better performance.
Moreover, this design provides a new perspective to generalize single-layer CS methods to multilayer settings by first applying the single-layer CS to each layer and then merging with \combinename. 
Additionally, the whole framework enjoys several desired properties, including 1) flexible layer participation. It allows layers to be included or excluded as needed; 2) independence from the order of layers. The order in which layers are processed does not impact the final results; 3) overlapping communities. The \singlesearchname~ naturally supports the search of overlapping communities, where nodes can belong to multiple communities simultaneously.

\vspace{1mm}
\noindent \textbf{Contributions.}
The main contributions are as follows:
\begin{itemize}[leftmargin=10pt, topsep=1pt]
\item{} We propose a new ensemble-based multilayer CS framework \framework, which operates without the need for labels. It consists of two core components: \singlesearchname~and \combinename.
\item{} \singlesearchname~ searches communities in each layer without labels. It uses a graph diffusion encoder to learn both layer-shared and layer-specific representations and identifies communities based on the inter-layer similarity-to-community module.
\item{} \combinename~obtains the consensus community by merging the community results from each layer using the EM algorithm. It iteratively estimates both the layer-specific error rate and community memberships with maximum likelihood estimation.
\item{} Experiments on 10 public datasets demonstrate the superiority of \framework~in terms of both accuracy and efficiency. Regarding accuracy, \framework~achieves an average of 8.78\%-52.18\% improvement in the F1-score compared to the existing multilayer CS methods. Additionally, \framework~offers competitive efficiency in both the searching and merging phases.
\end{itemize}

\section{Preliminaries}
\label{sec:Preliminary}

\begin{table}[t]
\centering 
\caption{Symbols and Descriptions}
\vspace{-0.4cm}
\label{tab:symbol}
\begin{tabular}{|p{2.6cm}|p{5.4cm}|}
\hline
\cellcolor{lightgray}\textbf{Notation} & \cellcolor{lightgray}\textbf{Description} \\ \hline\hline
$r_{max}, n$&the number of layers/nodes \\ \hline
$\mathcal{V}, \mathcal{E}$&the layers/edge set of all layers\\ \hline
$\mathcal{X}, \mathcal{A}$&feature/adjacency matrix set of all layers\\ \hline
$\mathcal{G}^{(r)}=\{\mathcal{V},\mathcal{E}^{(r)}\}$& the $r$-th layer of the multilayer graph\\ \hline
$X^{(r)}, A^{(r)}$& feature matrix and adjacency matrix of the $r$-th layer\\ \hline
$C^{(r)}, P^{(r)}$& layer-shared representation and layer-specific representation of the $r$-th layer\\ \hline
$U \in \mathbb{R}^{n \times f_d} $& the common variable \\ \hline
$q=\mathcal{V}_q$& the query with node set $\mathcal{V}_q$\\ \hline
$\mathcal{G}_q, \tilde{\mathcal{G}}_q$& ground-truth/predicted  community of $q$\\ \hline
$\mathcal{M}^{\theta}(\cdot)$ & neural network model with parameters $\theta$ \\ \hline
$S^{(r)}\in \mathbb{R}^{N}$& community score vector in $r$-th layer \\ \hline
$\mathcal{L}(\cdot)$ & training objective function \\ \hline

\end{tabular}
\vspace{-5mm}
\end{table}

In this section, we give relevant preliminaries and the introduction of state-of-the-art models for multilayer CS. 
\vspace{-3mm}
\subsection{Problem Statement}
We follow the common definition of multilayer graphs~\cite{behrouz2022firmtruss, luo2020local, interdonato2017local, mo2023disentangled}.
Let $\mathcal{G} = \{\mathcal{V}, \mathcal{E}\} = \{\mathcal{G}^{(1)}, \mathcal{G}^{(2)}, \ldots, \mathcal{G}^{(r_{max})}\}$ denote the multilayer graph, where $\mathcal{G}^{(r)} = \{\mathcal{V}, \mathcal{E}^{(r)}\}$ is the $r$-th layer of the multilayer graph, $r_{max}$ is the number of layers and $1\leq r\leq r_{max}$. $\mathcal{V} = \{v_1, v_2, \ldots, v_n\}$ and $\mathcal{E} = \{\mathcal{E}^{(1)}, \mathcal{E}^{(2)}, \ldots, \mathcal{E}^{(r_{max})}\}$ represent the node set and the edge set of all layers, respectively.  
$X^{(r)} \in \mathbb{R}^{n \times f}$ denotes original node features of the $r$-th layer with $n$ and $f$ characterizing the number of nodes and the dimension of node features, respectively.
$\mathcal{E}^{(r)}$ represent the edge set of the $r$-th layer. $A^{(r)} \in \mathbb{R}^{n \times n} $ denotes the graph structure of the $r$-th layer, where $A_{ij}^{(r)} = 1 $ iff $e^{(r)}_{ij} = (v_i, v_j) \in \mathcal{E}^{(r)} $. 
The proposed \framework~ learns the layer-shared representations $C^{(r)} \in \mathbb{R}^{n \times {f_h}} $ and layer-specific representations $P^{(r)} \in \mathbb{R}^{n \times {f_h}} $ through a common variable $U \in \mathbb{R}^{n \times {f_h}} $ in $r$-th layer, where ${f_h}$ is the dimension of features in the hidden space.
$S\in \mathbb{R}^{n}$ is used to denote the community scores which reflect the membership of each node.
We use $q$ and $\mathcal{V}_q$ interchangeably to denote the query node set. $\mathcal{G}_q$ and $\tilde{\mathcal{G}}_q$ are utilized to denote the ground-truth community and the predicted community {w.r.t.} $q$. Next, we give the formal definition of multilayer CS.

\begin{definition}{(Multilayer Community Search~\cite{behrouz2022firmtruss,luo2020local, interdonato2017local}).} 
Given a multilayer graph $\mathcal{G} = \{\mathcal{V}, \mathcal{E}\}$ and query $q$, the task of multilayer community search aims to identify a query-dependent connected subgraph (i.e., community) where nodes in the found community are densely intra-connected across layers.

\end{definition}

\vspace{-4mm}
\subsection{State-of-the-art Methods}
\framework~is closely related to both multilayer and single-layer CS approaches. To better highlight the differences between our proposed method and existing techniques, we provide a detailed introduction of CS-MLGCN\cite{behrouz2022cs}, a leading multilayer CS method, and TransZero, a state-of-the-art single-layer CS method.

\noindent \textbf{CS-MLGCN}. 

The model in CS-MLGCN, denoted as $\mathcal{M}^\theta(\cdot)$, produces a community score vector $S_q = \mathcal{M}^\theta(\mathcal{V}_q, \mathcal{G})\in \mathbb{R}^{n}$, which represents the membership likelihood of each node belonging to the predicted community {w.r.t.} the query $\mathcal{V}_q$. 
It consists of a query encoder and a view encoder which are based on GCNs to aggregate information in each layer and use the attention mechanism to combine information across layers.
To train the model, the Binary Cross Entropy (BCE) loss is used to measure the error between the predicted community score $S_q$ and the ground-truth vector $Y_q$, where $Y_{q,j} = 1$ if and only if node $v_j \in \mathcal{G}_q$. The $j$-th bit of $Y_q$, denoted as $Y_{q,j}$, indicates whether node $v_j$ is part of the community.
\begin{equation*} \mathcal{L}_q = \frac{1}{n} \sum_{i=1}^{n} -\left( Y_{q,i} \log S_{q,i} + \left( 1 - Y_{q,i} \right) \log \left( 1 - S_{q,i} \right) \right) \end{equation*}

After obtaining the community score, the learned model, along with the label-based threshold, is used to locate the final community. Specifically, CS-MLGCN calculates the community score through model inference and follows the Constrained Breadth-First Search algorithm from ~\cite{jiang2022query} for community identification. It expands outward by selecting neighbors with a community score greater than the label-based threshold.

However, CS-MLGCN relies on labels for both learning community scores and determining the threshold for community identification. Furthermore, as the model is trained on pre-selected communities, it tends to memorize the training data rather than generalize, making it difficult to predict unseen communities.

\noindent\textbf{TransZero}.
TransZero frees the model from label information by introducing a two-phase strategy: the offline pre-train phase and the online similarity-to-community phase.
In the offline phase, it pre-trains the community search graph transformer (CSGphormer) based on two label-free objectives, i.e., the personalization loss and the link loss, to train the CSGphormer and encode the graph topology and community information into the latent space.
The personalization loss aims to bring the representation of the center node and its corresponding community closer together and push away the representation of the center node from the communities of other nodes. 
The link loss aims to enhance the similarity between the representations of neighboring nodes.
In the online phase, it computes the community score based on the similarity of representations and then identifies the community with a new modularity-like cohesiveness function.

However, TransZero is designed for single-layer CS and cannot be directly applied to multilayer graphs. Its search process fails to incorporate inter-layer complementary information, which is crucial for effective multilayer CS.

\noindent\textbf{TransZero-CC}. Here, we introduce the details of TransZero-CC which generalizes TransZero to the multilayer setting by the co-association-based consensus clustering~\cite{tagarelli2017ensemble}. More methods, including TransZero-Vote which is based on voting strategy and TransZero-ME which is based on the mixture of ensemble~\cite{topchy2004mixture, wang2022stable}, are evaluated in Section~\ref{sec:Experiment}.
The overall strategy follows a "search and merge" strategy, where communities are first identified in each layer using the single-layer CS method, and then merged into a single consensus community.
The full algorithm is outlined in Algorithm~\ref{algo:naive_approach}. Given a multilayer graph, query nodes and a single-layer CS method as input, TransZero-CC first searches for communities in each layer by the single-layer method (Lines 1 to 3). 
Then, in the merge phase, it first constructs a co-association matrix $B$ based on the found communities. The value of $B$ is assigned to the number of layers in which the corresponding nodes are assigned the same community and then normalized by the number of layers (Lines 4 to 7).
The obtained co-association matrix is utilized as the edge weight of the new graph with threshold $\mu$, and the single-layer CS method is applied to the new constructed graph (Line 8). Finally, the identified community is returned (Line 10).

However, while this simple approach can generalize either single-layer CS, it fails to incorporate inter-layer complementary information during both the search and merge phases, resulting in limited performance. Our \framework~follows this strategy and improves both the search phase and the merge phase by proposing \singlesearchname~and \combinename, respectively.

\begin{algorithm}[t]
\captionsetup{singlelinecheck=false} 
\captionsetup{margin={0pt,5em}}
\caption{\parbox{\linewidth}{TransZero-CC}}
\label{algo:naive_approach}
\LinesNumbered
\DontPrintSemicolon
\KwIn{The multilayer graph $\mathcal{G}$, query $\mathcal{V}_q$, single layer CS model TransZero $\mathcal{M}^\theta(\cdot)$, thereshold $\mu$.}
\KwOut{The identified community $\Tilde{\mathcal{G}}_q$.}

\tcp*[l]{%
  \begin{tabular}{@{}l@{\hspace{0em}}r@{}}
    {\color{teal}\text{Search Phase.}} & {\color{teal}$\annot{Propose \textbf{\singlesearchname}~in \framework}$}
  \end{tabular}
}
\For{each $\mathcal{G}^{(r)} \in \mathcal{G}$}{
    $\mathcal{M}^{\theta} \leftarrow \text{Offline-Pre-train}(\mathcal{V}_q, \mathcal{G}^{(r)})$\;
    $\Tilde{\mathcal{G}}_q^{(r)} \leftarrow \text{Online-Search}(\mathcal{G}^{(r)}, \mathcal{M}^{\theta})$
}

\tcp*[l]{%
  \begin{tabular}{@{}l@{\hspace{0em}}r@{}}
    {\color{teal}\text{Merge Phase.}} & {\color{teal}$\annot{Propose \textbf{\combinename}~in \framework}$}
  \end{tabular}
}

$B \in \mathbb{R}^{n\times n} \leftarrow\text{Initialize co-association matrix} $\;
    \For{$i,j \leq n$}{
        $b_{ij} = \{r|\exists \Tilde{\mathcal{G}}_q^{(r)} \wedge (v_i, v_j)\in \mathcal{E}^{(r)}\}$\; 
        $B(i, j) = \frac{|b_{ij}|}{r_{max}}$\;
    }

$\Tilde{\mathcal{G}}_q\leftarrow \text{Apply TransZero on new graph } (\mathcal{V}, B) $ with $\mu$\;

\Return $\Tilde{\mathcal{G}}_q$\;

\end{algorithm}

\begin{figure*}
  \centering
  \includegraphics[width=0.995\linewidth]{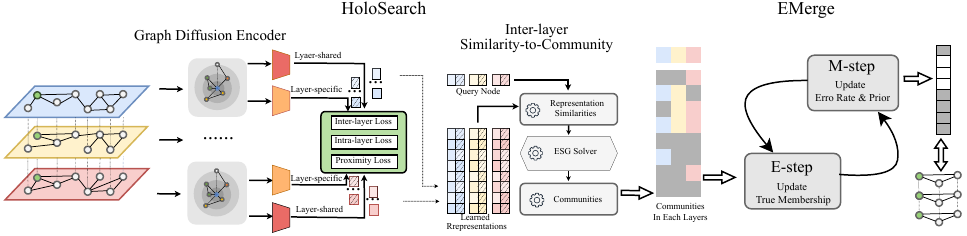}
  \vspace{-3mm}
  \caption{Illustration of the overall \framework~framework. }
  \label{fig:overview}
\vspace{-4mm}
\end{figure*}

\section{Overview}
\label{sec:Overview}

In this section, we provide an overview of \framework, which is designed for the multilayer CS. The overall architecture of \framework~is illustrated in Figure~\ref{fig:overview}. As shown in the figure, \framework~consists of two key components: \singlesearchname~and \combinename.
In the \singlesearchname~component, we propose a graph diffusion encoder that efficiently aggregates global and local information to learn both layer-shared and layer-specific representations. The training of the model is guided by three label-free loss functions: proximity loss to encourage closer representations of similar nodes, inter-layer loss to ensure inter-layer consistency and intra-layer loss to enhance discriminability.  
The learned representations are then used for inter-layer similarity-to-community to identify communities.
The \combinename~component is designed to merge the communities from each layer while accounting for the layer-specific properties of each layer using an EM algorithm. It consists of both the E-step and the M-step and iteratively applies these two steps to estimate both the true node assignments and error rates for each layer.
The details of \singlesearchname~are discussed in Section~\ref{sec:singlesearch}, and \combinename~ is detailed in Section~\ref{sec:EMerge}.

\vspace{-2mm}
\section{\singlesearchname}
\label{sec:singlesearch}

\framework~focuses on unsupervised multilayer CS, adopting a "search and merge" strategy through the introduction of \singlesearchname~and \combinename, respectively. Specifically, \singlesearchname~is designed to efficiently and effectively search for communities within each layer of a multilayer graph while leveraging both inter-layer and intra-layer information. This is achieved by integrating three key modules: 
1) the graph diffusion encoder that efficiently aggregates global-local information within each layer to learn both layer-shared representations and layer-specific representations; 
2) three label-free training objectives, including inter-layer loss, intra-layer loss and proximity loss, which aim to guide the model to encode cohesiveness priors and inter- and intra-layer correlations into the learned latent representations. These label-free objectives guide the model in capturing patterns critical for enhanced CS; 
3) the inter-layer similarity-to-community module that combines both layer-shared and layer-specific representations to efficiently and effectively identify the communities in each layer.

\vspace{-2mm}
\subsection{Graph Diffusion Encoder}
\noindent\textbf{Motivation.} 
To propagate information in the multilayer graph, a common approach in multilayer CS and multilayer learning is the use of graph neural networks, such as the GCN employed in CS-MLGCN~\cite{behrouz2022cs}. However, traditional GNNs heavily rely on local neighborhood information for message passing, often overlooking the broader global context. 
Moreover, traditional GNNs face efficiency and scalability challenges, particularly with large-scale multilayer graphs. While stacking multiple graph convolutional layers can expand the receptive field to improve accuracy, it often comes at the cost of computational bottlenecks and reduced scalability.
Additionally, nodes in multilayer graphs, influenced by global patterns spanning multiple layers and unique characteristics specific to each individual layer, are characterized by two complementary types of information: layer-shared information, which captures commonalities across layers, and layer-specific information, which captures features unique to each layer.
The integration of these dual representations enables a more holistic understanding of the relationships both across layers and within each layer.
The comprehensive integration of layer-shared representation and layer-specific representation allows the model to effectively integrate complementary information across layers, improving the accuracy and robustness of multilayer CS.

To more efficiently aggregate information, we introduce the graph diffusion encoder. 
The graph diffusion encoder comprises two main processes: the graph diffusion process and the feedforward process. 
The graph diffusion process leverages a diffusion mechanism to integrate features from both local neighbors and distant nodes. This approach enables a more comprehensive representation of the graph structure when learning both layer-shared and layer-specific representations.
The feedforward process utilizes a neural network to refine the aggregated features, enhancing their expressiveness and capturing non-linear relationships among the nodes.
The diffusion process is efficient, while the feedforward process is highly expressive. By first applying the graph diffusion process followed by the feedforward process, the encoder addresses the efficiency limitations of traditional GCNs while enhancing effectiveness for large-scale multilayer CS.

Specifically, the graph diffusion process utilizes the graph diffusion~\cite{gasteiger2019diffusion} to propagate features in each layer, as described below:

\begin{equation}
\begin{aligned}
     &O_r = \hat{\Delta}_r^{-\frac{1}{2}} \hat{A}^{(r)} \hat{\Delta}_r^{-\frac{1}{2}}, \\
     &O^*_r = \sum^{\infty}_{k=1}\theta_k O^k_r, \\
     &H^{(r)} = \hat{\Delta}_r^{-1} O^*_r \hat{X}^{(r)},
\end{aligned}  
\end{equation}
where $\hat{A}^{(r)} = \omega_{loop}I^{(r)}_N + A^{(r)} \in \mathbb{R}^{n \times n}$ is the augmented adjacency matrix, and $\omega_{loop}$ denotes the weight of self-loops which is set to 1 according to ~\cite{gasteiger2019diffusion}. 
$\hat{\Delta}_r$ is the degree matrix of $\hat{A}^{(r)}$, and \textit{k} denotes the \textit{k}-th diffusion. $\hat{\Delta}_r^{-1}$ is used to normalize the out-degree distribution of nodes, ensuring that the total amount of information that each node transmits to others remains consistent. 
For $\theta_k$, we employ the heat kernel ~\cite{kondor2002diffusion} to ensure the convergence of $O^*_r$. Specifically $\theta_k = e^{-t}\frac{t^k}{k!}$, where $t$ is used to denote the diffusion time. Following~\cite{gasteiger2019diffusion}, the diffusion process for a node terminates when its corresponding value in $\theta$ falls below the threshold of $0.0001$. Thus, the process does not involve infinite iterations but instead completes within a finite number of steps.

The feedforward process employs two separate Feedforward Networks (FFNs) with unshared parameters to map the node representations $H^{(r)}$ of each layer into layer-shared and layer-specific representations (i.e., $C^{(r)} \in \mathbb{R}^{n \times {f_h}}$ and $P^{(r)} \in \mathbb{R}^{n \times {f_h}}$ where ${f_h}$ is the dimension of the hidden representation) as follows. The FFNs~\cite{wang2024efficient} consist of an input layer, one or more hidden layers (one in our model) and an output layer, with each layer composed of neurons connected to all nodes in the adjacent layers.

\begin{equation}
\begin{aligned}
     &C^{(r)} = \text{FFN}(H^{(r)}) \\
     &P^{(r)} = \text{FFN}(H^{(r)})
\end{aligned}  
\end{equation}
The graph diffusion encoder directly models long-range dependencies in a single step, avoiding the need for multiple stacked layers that other graph convolutional networks require to propagate information across distant nodes, and thus achieves a faster speed.
Unlike traditional encoders that rely solely on direct node connections, the graph diffusion encoder diffuses information across multiple hops, balancing the influence of local neighborhood structure and global graph topology. This approach allows the model to more effectively capture long-range dependencies, leading to improved performance in the task of multilayer CS.

\begin{lemma}
\label{lemma: heat_kernel}
The heat kernel ensures the convergence.
\end{lemma}

\noindent\textbf{Proof of Lemma 1.}
For $\theta_k = e^{-t}\frac{t^k}{k!}, t \in \mathbb{R}^*$ is a hyper-parameter, $k \in \mathbb{N}^*$, our goal is to prove that as $k \rightarrow \infty, \theta_k \rightarrow 0$. Obviously $e^{-t}$ is a constant, so we can simplify the problem to as $k \rightarrow \infty, \frac{t^k}{k!} \rightarrow 0$. First, we apply Stirling's Approximation ~\cite{stirling1730methodus}, Stirling's approximation is an approximation for factorials, expressed as: $x! \approx \sqrt{2 \pi x} \left( \frac{x}{e} \right)^x$, which becomes increasingly accurate as $x$ approaches infinity. Second, substituting Stirling's approximation into our target expression, we obtain
$\frac{t^k}{k!} \approx \frac{t^k}{\sqrt{2 \pi k} \left( \frac{k}{t} \right)^k} = \frac{t^k t^k}{\sqrt{2 \pi k} k^k} = \frac{t^{2k}}{\sqrt{2 \pi k} k^k}$. Third, to analyze the expression more easily, we take the logarithm 
$L = \log_t \frac{t^{2k}}{\sqrt{2 \pi k} k^k} = 2k - \frac{1}{2} \log_t (2 \pi k) - k \log_t k$. Fourth, we calculate the limit of $L$, in the case of $t \in (1, \infty)$, as $k \rightarrow \infty$,
$L \sim 2k - k \log_t k \; (\text{neglecting smaller terms})$. Thus, as $k \rightarrow \infty$,
$L \sim 2k - k \log_t k \to -\infty$. This means that 
$\frac{t^{2k}}{\sqrt{2 \pi k} k^k} = t^L \to 0$. Therefore, we can conclude that in the case of $t \in (1, \infty)$, as $k \rightarrow \infty, \theta_k \rightarrow 0$.

\subsection{Training Objectives}
\label{sec: training_objective}
\noindent\textbf{Motivation.} 
In the context of unsupervised multilayer CS, we design three label-free training objectives that capture the cohesive prior and inter- and intra-layer information in the latent space following~\cite{wang2024efficient, mo2023disentangled, wang2024disensemi}: 1) Inter-layer loss $\mathcal{L}_{inter}$ to encourage layer-shared representations across layers to be as similar as possible; 2) Intra-layer loss $\mathcal{L}_{intra}$ to ensure the layer-shared and layer-specific representations of the same node are as distinct as possible in each layer; 3) Proximity loss $L_{p}$ to model the relationship between nodes and their communities by pulling the representation of a central node closer to its neighbors and pushing it away from other nodes. It is based on the observation that nodes rely on their communities to learn their unique representations in the graph.

\noindent\textbf{Inter-layer loss}. The inter-layer loss is computed as follows:

\begin{equation}
\begin{aligned}
\label{eqn:matching_loss}
    &{U} = \textrm{FFN}(\textrm{CONCAT}({C}^{(1)}, {C}^{(2)}, \ldots, {C}^{(r_{max})})), \\
    &\mathcal{L}_{inter} = \frac{1}{n} \sum_{r =1}^{r_{max}} \sum_{v \in \mathcal{V}} ({C}^{(r)}_v - {U}_v)^2
\end{aligned}
\end{equation}
Here, $\textrm{CONCAT}$ is an operator that concatenates the inputs, and $U \in \mathbb{R}^{n \times {f_h}}$ represents a shared latent space that aggregates the layer-shared information from each layer. The inter-layer loss $\mathcal{L}_{inter}$, as defined in Eqn~\ref{eqn:matching_loss}, encourages the alignment of layer-shared representations across different layers by minimizing the discrepancy between each representation and $U$. As training progresses, the model aims to converge toward a consistent representation across layers, i.e., ${C}^{(1)} \approx \ldots \approx {U} \approx \ldots \approx {C}^{(r_{max})}$, effectively capturing the shared structural patterns and interdependencies present in the multilayer graph. This consistency ensures that the learned representations generalize well across layers, encoding inter-layer relationships.

\noindent\textbf{Intra-layer loss}. 
To decouple the layer-shared and layer-specific representations, it is essential to enforce statistical independence between them. Specifically, when the layer-shared and layer-specific representations are statistically independent, the following constraint holds for the representations: 
$ \mathbb{E}[\phi^{(r)}({C}^{(r)}) \psi^{(r)}({P}^{(r)})] = \mathbb{E}[\phi^{(r)}({C}^{(r)})] \mathbb{E}[\psi^{(r)}({P^{(r)}})] $, and vice versa, where $\phi^{(r)}$ and $\psi^{(r)}$ are a linear layer, as demonstrated in~\cite{gretton2005kernel}, to project the representations into the same hidden space. This property implies that the statistical interactions between common and private representations are minimized when their correlation is minimized.

To achieve this decoupling, we minimize the correlation between $\phi^{(r)}({C}^{(r)})$ and $\psi^{(r)}({P}^{(r)})$. In particular, we adopt an intra-layer loss, computed as Pearson’s correlation coefficient between these two representations.  Formally, the intra-layer loss is computed by:
\begin{equation}
\label{eqn:correlation}
    \mathcal{L}_{{intra}} = \sum_{r=1}^{r_{max}} \frac{|\text{Cov}(\phi^{(r)}({C}^{(r)}), \psi^{(r)}({P}^{(r)}))|} {\sqrt{\text{Var}(\phi^{(r)}({C}^{(r)}))} \sqrt{\text{Var}(\psi^{(r)}({P}^{(r)}))}},
\end{equation}
where Cov(.,.) and Var(.) indicate covariance and variance operations, respectively. In Eqn~\ref{eqn:correlation}, the correlation coefficient between $\phi^{(r)}({C}^{(r)})$ and $\psi^{(r)}({P^{(r)}})$ is encouraged to converge to 0. Actually, as the correlation loss converges, the layer-shared representations ${C}^{(r)}$ and the layer-specific representations $ {P}^{(r)} $ are statistically independent.
By reducing the correlation, we encourage independence between ${C}^{(r)}$ and ${P}^{(r)}$, allowing each to capture distinct information. This ensures that the layer-shared representation focuses on shared structures across layers, while the layer-specific representation captures unique information within each layer.

\noindent\textbf{Proximity loss}. Intuitively, nodes depend on their communities to learn distinct representations within the graph. The proximity loss is designed to pull the representation of the central node closer to its neighbors while pushing it away from unrelated nodes. To achieve this, we first apply FFNs to project the neighborhood information into the latent space.

\begin{equation}
\label{eqn:hops}
    ^iZ^{(r)}= \textrm{FFN}_i\left(\textrm{CONCAT}(O^0_r \hat{X}^{(r)}, O^1_r \hat{X}^{(r)},\ldots, O^i_r \hat{X}^{(r)})\right)
\end{equation}
\begin{equation}
    ^{com}Z^{(r)} = \sum_{k=1}^{k_{max}} \alpha_k~^kZ^{(r)};
    \alpha_{} = \frac{\textrm{exp}((C^{(r)}|| ^kZ^{(r)})W_a)}{\sum_{i=1}^{k_{max}}\textrm{exp}((C^{(r)}|| ^iZ^{(r)})W_a)}
\end{equation}
Where ${k_{max}}$ is a hyper-parameter to control the scope (hops) of the neighborhood. We set it to 3 according to the experiments. $||$ is the concatenation operator, and $W_a \in \mathbb{R}^{ 2{f_h} \times 1}$ is the learnable weight matrix. Then, the proximity loss is modeled as margin triplet loss~\cite{schroff2015facenet} which is computed as follows.

\begin{equation}
    \mathcal{L}_p = 
    \sum_{r = 1}^{r_{max}}\mathcal{L}_p^{(r)}
    = \sum_{r = 1}^{R}\frac{1}{N^2}\sum_{v \in \mathcal{V}}\sum_{u \in \mathcal{V}}\mathcal{L}_p^{(r)}(u,v)
\end{equation}

\begin{equation}
    \mathcal{L}_p^{(r)}(u,v) = 
    -\textrm{max}\left(\sigma(C^{(r)}_v\ ^{com}Z^{(r)}_v) - \sigma (C^{(r)}_v\ ^{com} Z^{(r)}_u) + \epsilon,\ 0\right)
\end{equation}
where $\epsilon$ is the margin value and $\sigma(\cdot)$ is the sigmoid function.

\singlesearchname~takes the above three losses into account together. The overall loss function is defined as:

\begin{equation}
\label{equ:loss_func}
    \mathcal{L} = \mathcal{L}_p+\alpha\mathcal{L}_{mat}+\beta\mathcal{L}_{cor}
\end{equation}
where $\alpha, \beta \in [0,1]$ are the coefficients to balance the two losses. $\alpha$ and $\beta$ are set as 0.8 and 0.4 according to our experiments.
Note that all the $\mathcal{L}_p$, $\mathcal{L}_{mat}$ and $\mathcal{L}_{cor}$ do not contain information of ground-truth communities, and thus \singlesearchname~ do not rely on labels.

\subsection{Inter-layer Similarity-to-Community}

\begin{algorithm}[t]
\captionsetup{singlelinecheck=false} 
\captionsetup{margin={0pt,5em}}
\caption{\parbox{\linewidth}{Inter-layer Similarity-to-Community}}
\label{algo:score_computation}
\LinesNumbered
\DontPrintSemicolon
\KwIn{The query $\mathcal{V}_q$, multilayer graph $\mathcal{G}=\{\mathcal{V}, \mathcal{E}\}$, layer-shared representations $C$ and layer-specific representations $P$, the balancing coeffient $\lambda$.}
\KwOut{The identified community $\Tilde{\mathcal{G}}_q$.}

Initialize $^cS^{(r)},\ ^pS^{(r)}\leftarrow \{^cs^{(r)}_v =0,\ ^ps^{(r)}_v =0\ \textrm{for}\ v \in \mathcal{V}\}$\;
\For{each layer $r \in \{1,\cdots,r_{max}\}$}{
\For{$\{v\} \in \mathcal{V}$}{
     \For{$\{u\} \in V_q$}{
        $^cs_v^{(r)} =\ ^cs_v^{(r)}+ \frac{\sum_{i=0}^{f_h}C_{v,i}C_{u,i}}{\sqrt{\sum_{i=0}^{f_h}C_{v,i}C_{v,i}}\times \sqrt{\sum_{i=0}^{f_h}C_{u,i}C_{u,i}}}$\;

        $^ps_v^{(r)} =\ ^p_vs^{(r)}+ \frac{\sum_{i=0}^{f_h}P_{v,i}P_{u,i}}{\sqrt{\sum_{i=0}^{f_h}P_{v,i}P_{v,i}}\times \sqrt{\sum_{i=0}^{f_h}P_{u,i}P_{u,i}}}$\;
}
$^cs_v^{(r)}\leftarrow\frac{^cs_v^{(r)}}{|V_q|} $; $^ps_v^{(r)}\leftarrow\frac{^ps_v^{(r)}}{|V_q|} $;
}
$^c\mu^{(r)},\ ^c\sigma^{(r)}\leftarrow \textrm{compute mean and variance on } ^cS^{(r)}$ \;
$^p\mu^{(r)},\ ^p\sigma^{(r)}\leftarrow \textrm{compute mean and variance on } ^pS^{(r)}$ \;
$S^{(r)} = \frac{^cS^{(r)}-^c\mu^{(r)}}{^c\sigma^{(r)}}+\lambda \frac{^pS^{(r)}-^p\mu^{(r)}}{^p\sigma^{()}}$ \;
}
\textbf{Initialize} thread pool with $r_{max}$ threads\;
\For{each layer $r \in \{1,\cdots,r_{max}\}$ \textbf{in parallel} using threads}{
    $t_s= 0$; $t_e = \left|{S^{(r)}} \right|$\;
    $\hat{S}^{(r)} \leftarrow$ sort $S^{(r)}$ from large to small\;
    \While{$t_s \textless t_e$}{
        $\mathcal{V}_{mid} = \{v_i| \hat{s}_i \geq \hat{s}_{\frac{t_s+t_e}{2}}\}$; $\mathcal{V}_{left} = \{v_i| \hat{s}_i \geq \hat{s}_{\frac{t_s+t_e}{2}-1}\}$\;
        \If{$\textrm{ESG}(\hat{S}^{(r)}, \mathcal{V}_{mid}, \mathcal{G}) > \textrm{ESG}(\hat{S}^{(r)}, \mathcal{V}_{left}, \mathcal{G})$}{
             $t_s \leftarrow \frac{t_s+t_e}{2}$\;
        }
        \Else{
             $t_e \leftarrow \frac{t_s+t_e}{2}$\;
        }
    }
    \Return $\Tilde{\mathcal{G}}_q^{(r)} \leftarrow \textrm{induced subgraph of }\mathcal{V}_q \cup \{v_i| \hat{s}_i \geq \hat{s}_{t_e}\}$\;
}


\end{algorithm}

After pre-training, the similarity of the representation of two nodes can serve as a measure of their proximity. TransZero proposes a similarity-to-community to identify the community based on the similarity of representations without the usage of labels. However, it lacks the inter-layer information. Here, we propose the inter-layer similarity-to-community.
We first calculate the community score by computing the cosine similarity between the representations of the query node and other nodes in each layer. Alternative similarity metrics, such as L1 and L2 similarity, are also evaluated in the experiments in Section~\ref{sec:Experiment}. We compute both the layer-shared community scores and the layer-specific community scores. 
We introduce a hyperparameter, $\lambda$, to integrate inter- and intra-layer information and balance the contributions of these two community scores. Based on our experiments, $\lambda$ is set to -1, encouraging the selection of communities with higher similarity in layer-shared representations and lower similarity in layer-specific representations.

After obtaining the layerwise community score that incorporates both inter- and intra-layer information, the next step is to identify the community based on these scores. Strategies like a fixed threshold struggle to adapt to the diverse nature of real-world communities. Thus, we follow the expected score gain (ESG) objective proposed in ~\cite{wang2024efficient}, which aims to identify communities with high ESG. The definition of ESG is as follows.

Given a community found in $r$-th layer $\mathcal{G}_q^{(r)}=(\mathcal{V}_{\mathcal{G}_q^{(r)}}, \mathcal{E}_{\mathcal{G}_q^{(r)}})$ and community score $S^{(r)}$, the expected score gain of $\mathcal{G}_q^{(r)}$ is: 
\begin{equation}
    \label{equ:density_modularity}
    \textrm{ESG}(S^{(r)}, \mathcal{V}_{\mathcal{G}_q^{(r)}}, \mathcal{G})=\frac{1}{|\mathcal{V}_{\mathcal{G}_q^{(r)}}|^\tau}(\sum_{v\in \mathcal{V}_{\mathcal{G}_q^{(r)}}}s_v^{(r)}-\frac{\sum_{u\in \mathcal{V}}s_u^{(r)}}{|\mathcal{V}|}|\mathcal{V}_{\mathcal{G}_q^{(r)}}|)
\end{equation}
Here, $\tau$ is a hyperparameter that controls the granularity of community and is set as 0.9 guided by experiments in Section~\ref{sec:Experiment}. The first term in the brackets sums up the scores within the community, while the second term represents the average scores in the entire graph. ESG measures the extent to which the community score is more cohesive than the overall graph. Therefore, we aim to select nodes that form a community with an ESG value as high as possible. 

\begin{lemma}
\label{lemma: submodular}
The ESG function is not submodular.
\end{lemma}

\begin{proof}
    We prove this lemma by constructing a case that violates the definition of the submodular function~\cite{fujishige2005submodular}. Let $\tau=0.9$
    and consider two sets $A=\{x\}$, $B=\{x, y\}$ and node $u$ where $x,y,u\in V$ and $s_x=1.0, s_y=0.5, s_u=0.1$. $ESG(\{u\} \cup A)-ESG(A)=-0.0182-0.4667=-0.4849$ and $ESG(\{u\} \cup B)-ESG(B)=0-0.1964=-0.1964$. Therefore, we have $ESG(\{u\} \cup A)-ESG(A) =-0.4849 \textless ESG(\{u\} \cup B)-ESG(B) =-0.1964$ which violates the definition of submodular function. Thus, the ESG function is not submodular.
\end{proof}

The overall inter-layer similarity-to-community module is shown in Algorithm~\ref{algo:score_computation}. This algorithm takes as inputs the set of query nodes $\mathcal{V}_q$, the multilayer graph $\mathcal{G}$, the learned layer-shared representations $C$, the learned layer-specific representation $P$ and the balancing coefficient $\lambda$. The algorithm proceeds by first initializing both the layer-shared community scores and layer-specific community scores to zeros for all nodes (Line 1). Then, in each layer, for each node in the graph, the algorithm calculates the cosine similarity between the node and the query nodes using both the layer-shared and layer-specific representations to take into account inter-layer information into account (Lines 2 to 10). These similarity scores are accumulated across all query nodes and normalized by the number of query nodes  (Line 7).
We observe that  $^cS_{(r)}$ is significantly greater than $^pS^{(r)}$. This disparity makes the layer-specific features ineffective in searches. To mitigate this issue, normalization is applied separately to each node of $^cS_{(r)}$ and $^pS^{(r)}$. 
We compute the mean and variance of the accumulated scores across all nodes, for both the layer-shared and layer-specific representations. The final community score for each node is then calculated by normalizing the layer-shared and layer-specific scores using their respective mean and variance, and linearly combining them with the balancing coefficient $\lambda$ (Lines 8 to 10). 

Moreover, as highlighted in~\cite{wang2024efficient}, the ESG value for the first $k$ nodes in the sorted queue of community scores initially increases and then decreases as $k$ increases. This property enables us to utilize the binary search to find the community with the global highest ESG value. Based on this property, we parallel search the community in each layer using the binary search strategy.
First, a thread pool is initialized with $r_{max}$ threads (Line 11), where each thread handles the computation for one layer. For each layer $r$, variables are initialized and the community score $S^{(r)}$ is sorted in descending order (Lines 13 and 14), and the binary search approach is employed to iteratively refine the set of candidate nodes: $\mathcal{V}_{mid}$ and $\mathcal{V}_{left}$ (Lines 15 to 20). Based on the ESG comparison, the algorithm narrows the search range by adjusting the start and end pointers, $t_s$ and $t_e$.
Once the optimal threshold of the community score is determined, the algorithm returns the induced subgraph $\Tilde{\mathcal{G}}_q^{(r)}$ (Line 21), consisting of the query nodes $\mathcal{V}_q$ and the nodes whose community scores exceed the final threshold.

\section{\combinename~Algorithm}
\label{sec:EMerge}

After identifying the community in each layer, the key challenge is how to merge these communities from different layers into a single, consensus community. As discussed in Section~\ref{sec:Introduction}, naive approaches such as simple voting only aggregate the final results from each layer, without accounting for the layer-specific characteristics of each layer or the varying response probabilities that layers assign to different node categories. 
In this section, we introduce the \combinename~algorithm, an Expectation-Maximization (EM)-based merge method inspired by\cite{dawid1979maximum}. This algorithm iteratively estimates both the layer-specific uncertainties (i.e., the error rates of each layer) and the true community memberships, thereby achieving a more accurate consensus community.

The EM algorithm is an iterative method for finding maximum likelihood estimates of parameters in probabilistic models with latent unobserved variables. The algorithm consists of two steps: E-step (Expectation), calculating the expected value of the latent variables; M-step (Maximization), re-estimating the model parameters using the E-step results. Starting from an initial parameter estimate, these two steps are repeated iteratively until convergence to a local maximum likelihood estimate. The EM algorithm is particularly useful for finding maximum likelihood estimates with latent unobserved data. Its simplicity and broad applicability have made EM algorithm a popular choice for problems involving latent unobserved variable models.

We therefore model both the layer-specific uncertainties (i.e., the error rates of each layer) and the true community memberships as latent variables, adapting the EM algorithm for community merging. By treating these variables as latent, the EM algorithm can iteratively optimize both based on the discovered communities. 
In this formulation, each layer provides a binary classification result (i.e., communities), determined using \singlesearchname, indicating whether each node belongs to the community.
\combinename~consists the initialize step, E-step and M-step. Assume that all the categories are denoted as $\mathcal{J}$ and the identified communities in all the layers are organized as ${D}\in\mathbb{R}^{n\times |\mathcal{J}|\times r_{max}}$. The term $D_{ij}^{(r)}$ represents the number of times node $v_i$ is assigned to category $j$ in the $r$-th layer. Note that, each layer evaluates each node $v_i$ only once, which means that $D_{ij}^{(r)}$ can take a value of either 0 or 1. Next, we will introduce the details of these steps.

\noindent\textbf{Initialize-step}. We begin by initializing the latent variable $T \in \mathbb{R}^{n \times 2}$, which indicates the true membership probabilities of each node in one of two possible categories: either belonging to the community (category 1) or not (category 0). Specifically, $T_{ij}$ denotes the probability that the $i$-th node is assigned to category $j$. $T_{ij}$ in 0-th iteration are computed as: 
\begin{equation}
\label{eqn:initial}
    ^{[0]}T_{ij} = \frac{\sum_{r} D_{ij}^{(r)}}{\sum_{r} \sum_{l} D_{il}^{(r)}},
\end{equation}
where $j \in \{0, 1\}$ and $l \in \{0, 1\}$ represent the two categories, where 0 indicates non-membership and 1 indicates membership in the community. The initialization of $T$ is carried out by averaging the classification decisions across all layers. 
Following initialization, $T$ is iteratively optimized by accounting for both the error rate and the prior tendencies of each layer in classifying nodes in the E-step and the M-step. 

\noindent\textbf{M-step}. In the M-step, given the true membership, we compute $\pi$, the error rate of each layer and $\eta$, prior of each category. The proof can be found in Section ~\ref{sec:Analysis}.
\begin{equation}
\begin{aligned}
\label{eqn:E_step}
    &^{[m]}\pi_{jl}^{(r)} = {\sum_i}^{[m-1]}T_{ij}D_{il}^{(r)} \bigg/ \sum_l {\sum_i}^{[m-1]}T_{ij}D_{il}^{(r)}, \\
    &^{[m]}{\eta}_j = {\sum_i}^{[m-1]}T_{ij} / n,
\end{aligned}
\end{equation}
where $^{[m]}\pi_{jl}^{(r)}$ represents the probability in the m-th iteration that the $r$-th layer mistakenly classifies a node from category $j$ as category $l$, with $\pi^{(r)} \in \mathbb{R}^{2 \times 2}$. Additionally, $^{[m]}\eta_j \in \mathbb{R}^{|\mathcal{J}|}$ denotes the occurrence probability of category $j$ in m-th iteration. The expression $^{[m]}T_{ij}D_{il}^{(r)}$ indicates that the true membership of node $v_i$ is $j$ (as reflected by $^{[m]}T_{ij}$), while the $r$-th layer classifies it as category $l$ (as indicated by $D_{il}^{(r)}$).

\noindent\textbf{E-step}. The E-step updates the true membership given the error-rate of each layer and the prior of each category computed in the E-step. According to the Bayesian method~\cite{joyce2003bayes}, the relationship between $T_{ij}$, $p_j$ and $\pi$ is as follows.

\begin{equation}
\begin{aligned}
&p(T_{ij} = 1 | {D}) \propto p({D} | T_{ij} = 1)p(T_{ij} = 1) \propto \prod_{r=1}^{r_{max}} \prod_{l=0}^{|\mathcal{J}|} (\pi_{jl}^{(r)})^{D_{il}^{(r)}} \eta_{j}, \\   
\end{aligned}
\end{equation}
Here, $(\pi_{jl}^{(r)})^{D_{il}^{(r)}}$ indicates that probability that $r$-th layer classify category $j$ into $l$ (as indicated by $\pi_{jl}^{(r)}$) while the original decision by $r$-th layer is category $l$ (as indicated by ${D_{il}^{(r)}}$). According to above formulation, we can update the true membership as follows which normalize the true membership matrix by row.
\begin{equation}
\label{eqn:M_step}
    p(^{[m+1]}T_{ij} = 1 | {D}) = \frac{\prod_{r=1}^{r_{max}} \prod_{l=0}^{|\mathcal{J}|} (^{[m]}\pi_{jl}^{(r)})^{D_{il}^{(r)}} \ ^{[m]}\eta_{j}}
{\sum_{q=0}^{|\mathcal{J}|} \prod_{r=1}^{r_{max}} \prod_{l=1}^{|\mathcal{J}|} (^{[m]}\pi_{ql}^{(r)})^{D_{il}^{(r)}} \ ^{[m]}\eta_{q}}, 
\end{equation}

\begin{algorithm}[t]
\captionsetup{singlelinecheck=false}
\captionsetup{margin={0pt,5em}}
\caption{\parbox{\linewidth}{\combinename~Algorithm}}
\label{algo:em_observer_error}
\LinesNumbered
\DontPrintSemicolon
\KwIn{The community decision data $D$, tolerance $\epsilon$.}
\KwOut{The identified true membership probabilities $T$.}
Initialize ${\!^{[0]}T}$ according to Eqn~\ref{eqn:initial}\;
\For{$m = 1, 2, \ldots$ until convergence}{
\For{all $j, l, r$}{
$^{[m]}\pi_{jl}^{(r)} = \sum_i\ ^{[m-1]}T_{ij}D_{il}^{(r)} \bigg/ \sum_l \sum_i T_{ij}D_{il}^{(r)}$\;
$^{[m]}{p}_j = \sum_i\ ^{[m-1]}T_{ij} / n$
}
\For{all $i, j$}{
$p(^{[m+1]}T_{ij} = 1 | {D}) = \frac{\prod_{r=1}^{r_{max}} \prod_{l=1}^{|\mathcal{J}|} (^{[m]}\pi_{jl}^{(r)})^{D_{il}^{(r)}} \ ^{[m]}p_{j}}
{\sum_{q=1}^{|\mathcal{J}|} \prod_{r=1}^{r_{max}} \prod_{l=1}^{|\mathcal{J}|} (^{[m]}\pi_{ql}^{(r)})^{D_{il}^{(r)}} \ ^{[m]}p_{q}}$;
}
\If{$|^{[m+1]}T-^{[m]}T|$ < $\epsilon$}{
             Break\;
        }
}
\Return $\text{SOFTMAX}(^{[m]}T)$;
\end{algorithm}

\noindent\textbf{Overall algorithms}. 
The overall algorithm of \combinename~is summarized in Algorithm~\ref{algo:em_observer_error}. The algorithm takes as inputs the community decision data D which represents the communities identified by each layer for each node, and a convergence tolerance $\epsilon$. It outputs the estimated true membership probabilities T for each node. The process begins by initializing the membership probabilities using Eqn~\ref{eqn:initial} (Line 1). Then, it enters an iterative process (Lines 2 to 9). It repeatedly executes the M-step (Lines 3 to 5) and E-step (Lines 6 to 7) until the parameters converge (Lines 8). 
Upon convergence, this distribution T is then processed through the SOFTMAX function to determine the final categorical distribution of $\mathcal{V}$ (Line 10).

\section{Theoretical Analysis}
\label{sec:Analysis}

\noindent\textbf{The proof of Eqn~\ref{eqn:E_step}}.
The marginal log-likelihood of the observed data is as follows:
\begin{equation}
L(\theta; {D}, T) = \sum_i \sum_j T_{ij} [\log \eta_j + \sum_r \sum_l D_{il}^{(r)} \log \pi_{jl}^{(r)}]
\end{equation}
Where $\theta = \{\eta_j, \pi_{jl}^{(r)}\}$ are the parameters to estimate to estimate for the next E-step given current parameters $^{[m]}\theta$. The Q function can be computed as: 
\begin{align}
Q(\theta | ^{[m]}\theta) &= E[L(\theta; {D}, T) | {D},\! ^{[m]}\theta] \\
&= \sum_i \sum_j E[T_{ij} | {D},\!^{[m]}\theta] [\log \eta_j + \sum_r \sum_l D_{il}^{(r)} \log \pi_{jl}^{(r)}] \nonumber
\end{align}
Where $E[T_{ij} | {D},\! ^{[m]}\theta]$ is the $p(T_{ij} = 1 | {D})$

In the M-step, we need to maximize the Q function. This can be done by computing the derivative of the Q function with respect to the parameter $p_j$, and setting the gradient to zero.
\begin{equation}
\frac{\partial Q}{\partial \eta_j} = \sum_i \frac{E[T_{ij} | {D},\!^{[m]}\theta]}{\eta_j} - \lambda = 0
\end{equation} Where $\lambda$ is used to guarantee $\sum_j \eta_j = 1$. Solve this equation, we can get that:
\begin{equation}
\eta_j = \frac{\sum_i E[T_{ij} | {D}, \!^{[m]}\theta]}{n} = \frac{\sum_i T_{ij}}{n}
\end{equation}

Similarly, we can compute the derivative of the Q function w.r.t. the parameter $\pi_{jl}^{(r)}$, and get the updated estimation of $\pi_{jl}^{(r)}$.
\begin{equation}
\frac{\partial Q}{\partial \pi_{jl}^{(r)}} = \sum_i \frac{E[T_{ij} | {D}, \!^{[m]}\theta] D_{il}^{(r)}}{\pi_{jl}^{(r)}} - \lambda_{jr} = 0
\end{equation}

\begin{equation}
\pi_{jl}^{(r)} = \frac{\sum_i E[T_{ij} | {D}, \theta^{(t)}] D_{il}^{(r)}}{\sum_i E[T_{ij} | {D}, \theta^{(t)}] \sum_l D_{il}^{(r)}} = \frac{\sum_i\ T_{ij}D_{il}^{(r)}}{\sum_l \sum_i T_{ij}D_{il}^{(r)}}
\end{equation}

\noindent \textbf{Time complexity analysis.}  
The heat-kernel-based graph diffusion needs $O(n)$~\cite{gasteiger2019diffusion} to aggregate the information. The FFNs in the encoder need $O({f_h}\times {f_h}\times L)$ for FFNs with L layers of the fully connected layers. Thus, \singlesearchname~needs $O\left(\left(n+{f_h}\times {f_h}\times L\right)\times Epoch\right) $ for pre-training the model where $Epoch$ is the number of training epochs. The score computation needs $O(|V_q|\times n\times {f_h})$ for one layer and needs $O(|V_q|\times n \times {f_h} \times r_{max})$ for $r_{max}$ layers. The identification needs $O(n\times\log{n})$ for sorting the scores. It iterates at most $O(\log{n})$ iterations with each iteration needing
$O(n)$ operations. Therefore, \singlesearchname~needs $O(|V_q|\times n\times {f_h} \times r_{max}+2\times n \times \log{n}\times r_{max})$ for identification the community of one query.

The update of $\pi$ needs $O(2\times 2\times r_{max} \times n)$ operation as there are only two categories. The update of $p$ also needs $O(4\times r_{max}\times n)$. The update of $T$ needs $O(4\times r_{max}\times n)$. Assume that the \combinename~iterates for $I$ iterations, the time complexity of \combinename~ is $O(12\times r_{max}\times n\times I)$.

\begin{table}[t]
\centering 
\caption{The Profiles of Datasets}
\vspace{-0.4cm}
\label{tab:datasets}
\begin{tabular} 
{|p{1.3cm}<{\centering}|p{0.6cm}<{\centering}|p{0.7cm}<{\centering}|p{0.7cm}<{\centering}|p{1.0cm}<{\centering}|p{0.7cm}<{\centering}|p{1.0cm}<{\centering}|}

\hline
\cellcolor{lightgray}\textbf{Datasets} & \cellcolor{lightgray}\textbf{Alias} & \cellcolor{lightgray}\textbf{$|\mathcal{V}|$} & \cellcolor{lightgray}\textbf{$R$} & \cellcolor{lightgray}\textbf{$|\mathcal{E}|$} & \cellcolor{lightgray}\textbf{$|\mathcal{G}_q|$} & \cellcolor{lightgray}\textbf{$f$}\\ \hline \hline

Aucs              & AU & 61                       & 5  & 2.4K  & 9  & 61   \\\hline
Terrorist         & TR & 79                       & 14 & 2.2K  & 5  & 32   \\ \hline
RM                & RM & 91                       & 10 & 14K   & 2  & 32   \\ \hline
3sources          & 3S & 169                      & 3  & 7.5K  & 5  & 3631 \\ \hline
Wikipedia & WK & 693                      & 2  & 11.5K & 10 & 10   \\ \hline
ACM               & AC & 3025                     & 2  & 1.1M  & 3  & 1870 \\ \hline
Freebase          & FB & 3492                     & 3  & 142K  & 3  & 3492 \\ \hline
DBLP              & DB & 4057                     & 3  & 5.9M  & 4  & 334  \\ \hline
IMDB              & IM & 4780                     & 2  & 64K   & 3  & 1232 \\ \hline
Higgs             & HG & 456K & 4  & 13M   & 35 & 64  \\ \hline

\end{tabular}
\vspace{-5mm}
\end{table}
\section{Experimental Evaluation}
\label{sec:Experiment}

\begin{table*}[t]
\centering 
\caption{Overall F1-score performance}
\vspace{-0.4cm}
\label{tab:exp1_f1score}
\begin{tabular} {|p{1.9cm}<{\centering} p{2.2cm}<{\centering} | p{0.8cm}<{\centering} p{0.8cm}<{\centering} p{0.8cm} <{\centering} p{0.8cm}<{\centering} p{0.8cm}<{\centering} p{0.8cm}<{\centering} p{0.8cm}<{\centering} p{0.8cm}<{\centering} p{0.8cm}<{\centering} p{0.8cm}<{\centering} | p{1.2cm}<{\centering}|}
\hline
\textbf{Categories} & \textbf{Methods} & AU & TR & RM & 3S & WK & AC & FB  &  DB & IM & HG & Improve \\ \hline \hline
\multirow{4}{*}{\textbf{Criterion-based}} 
& FirmCore   &  0.2593 & 0.4615 & 0.7751 & 0.5319 & 0.2012 & 0.2728 & 0.1017 & 0.6706 & 0.0695 & 0.2052  & 49.85\%  \\ 
& FirmTruss   & 0.2917 & 0.4615 & \cellcolor[RGB]{255,200,200}\textbf{0.8905} & 0.5628 & 0.3387 & 0.3609 & 0.2313 & {0.7613} & 0.4345 & 0.2770  & 39.24\% \\ 
& FocusCore & 0.1958 & 0.5214 & 0.7610 & 0.5523 & 0.1802 & 0.3228 & 0.4882 & 0.4006 & 0.3743 & 0.0015 & 47.36\%  \\ 
\hline
\multirow{3}{*}{\textbf{Learning-based}} 
& RWM    & 0.2912 & 0.3776 & 0.4563 & 0.4326 & 0.2081 & 0.4569 & 0.3281 & 0.1385 & 0.3745 & 0.2522  & 52.18\%  \\ 
& CS-MLGCN(IN)   & 0.7089 & 0.3214 & 0.1005 & 0.2417 & 0.3404 & 0.2730 & 0.4891 & 0.4355 & 0.4269 & 0.3355 & 48.61\%   \\ 
& CS-MLGCN(TD)   & 0.7868 & \underline{0.7653} & \underline{0.8289} & 0.6100 & 0.6054 & 0.6505 & 0.4891 & 0.6924 & 0.4833 & 0.6809 & 19.41\%   \\ 
& CS-MLGCN(HB)    & 0.8034 & 0.4220 & 0.5593 & 0.5341 & 0.5925 & 0.4482 & 0.6139 & 0.5332 & 0.4833 & 0.6645 & 28.79\%   \\ 
\hline
\multirow{3}{*}{\textbf{Single-layer}} 
& TransZero-Vote   & {0.8235} & 0.7248 & 0.5578 & {0.7342} & {0.8796} & {0.6556} &	{0.7262} & 0.7494 & {0.7446} & {0.8230} & 11.15\%   \\ 
& TransZero-ME   & 0.8417 & 0.7392 & 0.5720 & \underline{0.7562} & 0.8927 & \underline{0.6895} & 0.7333 & 0.7604 & 0.7563 & 0.8329 & 9.60\%   \\ 
& TransZero-CC   & \underline{0.8539} & 0.7508 & 0.5792 & {0.7551} & \underline{0.8985} & {0.6810} & \underline{0.7411} & \underline{0.7719} & \underline{0.7647} & \underline{0.8594} & 8.78\%   \\ 
\hline
\multirow{1}{*}{\textbf{Our}} 
& \framework~   & \cellcolor[RGB]{255,200,200}\textbf{0.9214} & \cellcolor[RGB]{255,200,200}\textbf{0.8842} & 0.6343& \cellcolor[RGB]{255,200,200}\textbf{0.9011} & \cellcolor[RGB]{255,200,200}\textbf{0.9523} & \cellcolor[RGB]{255,200,200}\textbf{0.8007} & \cellcolor[RGB]{255,200,200}\textbf{0.7850} & \cellcolor[RGB]{255,200,200}\textbf{0.8615} & \cellcolor[RGB]{255,200,200}\textbf{0.8117} & \cellcolor[RGB]{255,200,200}\textbf{0.9822} & - \\ 
\hline


\end{tabular}
    \begin{flushleft}

        \footnotesize $^*$ \colorbox[RGB]{255,200,200}{\textbf{Bold red}} indicates the best result. \underline{Underlined} indicates the second-best result. The last column presents the average improvement by \framework.
    \end{flushleft} 
\vspace{-6mm}
\end{table*}

\vspace{-1mm}
\subsection{Dataset Description}
Following prior works~\cite{mo2023disentangled, behrouz2022firmtruss}, we evaluated our approach using 10 publicly available datasets, including the large-scale multilayer graphs. A summary of the dataset characteristics is provided in Table~\ref{tab:datasets}. In this table, $|\mathcal{V}|$ represents the number of nodes, $R$ denotes the number of layers in the multi-layer graph, $|\mathcal{E}|$ refers to the number of edges, $|\mathcal{G}_q|$ is the number of communities, and $f$ indicates the original feature dimension.

\vspace{-3mm}
\subsection{Experimental Setup}

\noindent\textbf{Baselines:}
We compare \framework~ with 6 existing multilayer CS methods, including both rule-based and learning-based methods. For rule-based methods, we compare 1) FirmCore~\cite{hashemi2022firmcore}, which extends \textit{k}-core to model community structure; 2) FirmTruss~\cite{behrouz2022firmtruss}, which extends \textit{k}-truss for community modeling; and 3) FocusCore~\cite{wang2024focuscore}, which extends \textit{k}-core model and find community by both focus constraint and background constrain; For the learning-based models, we compare with 4) RWM~\cite{luo2020local}, which uses random walkers in each layer to capture local proximity w.r.t. the query nodes, ultimately returning a subgraph with minimal conductance; 5) CS-MLGCN~\cite{behrouz2022cs}, a leading supervised method for Multilayer CS. We evaluate CS-MLGCN under inductive, transductive and hybrid settings, as detailed below; and 6) TransZero, a state-of-the-art unsupervised single-layer CS method, is extended with a voting strategy (denoted as TransZero-Vote), along with a mixture of ensemble (denoted as TransZero-ME) and co-association techniques (denoted as TransZero-CC), to support multilayer CS, as detailed in Section~\ref{sec:Preliminary}.  

\noindent\textbf{Query Generation:}
Following ~\cite{wang2024efficient}, we use the following three mechanisms to generate the queries (containing $1\sim 3$ nodes).

\begin{itemize}[leftmargin=10pt, topsep=1pt]
\item{} \textbf{Inductive Setting (shortened as ID).} We randomly split all ground-truth communities into training and testing groups ($\sim$1:1 ratio). Training and validation queries come from the training communities, while test queries come from the testing communities. This setup evaluates the ability of the model to predict unseen communities. 
\item{} \textbf{Transductive Setting (shortened as TD).} We generate all the queries randomly from all the ground-truth communities.
\item{} \textbf{Hybrid Setting (shortened as HB).} We randomly split ground-truth communities into training and testing groups (~1:1 ratio). Training and validation queries come from the training set, while test queries are drawn from all communities.
\end{itemize}
Note that only the label-based method (i.e., CS-MLGCN) is trained under all three settings, while other methods, being label-free, produce consistent results across the three settings.

\noindent\textbf{Metrics:} Following previous methods for multilayer CS~\cite{hashemi2022firmcore, behrouz2022firmtruss, behrouz2022cs}, we use the F1-score to assess the accuracy of the identified communities relative to the ground-truth communities. A higher F1-score indicates a better performance.

\begin{figure*}
\subfigcapskip=-7pt 
    
    \subfigure[Efficiency results of the training phase]{ 
        
        \includegraphics[width=0.485\textwidth]{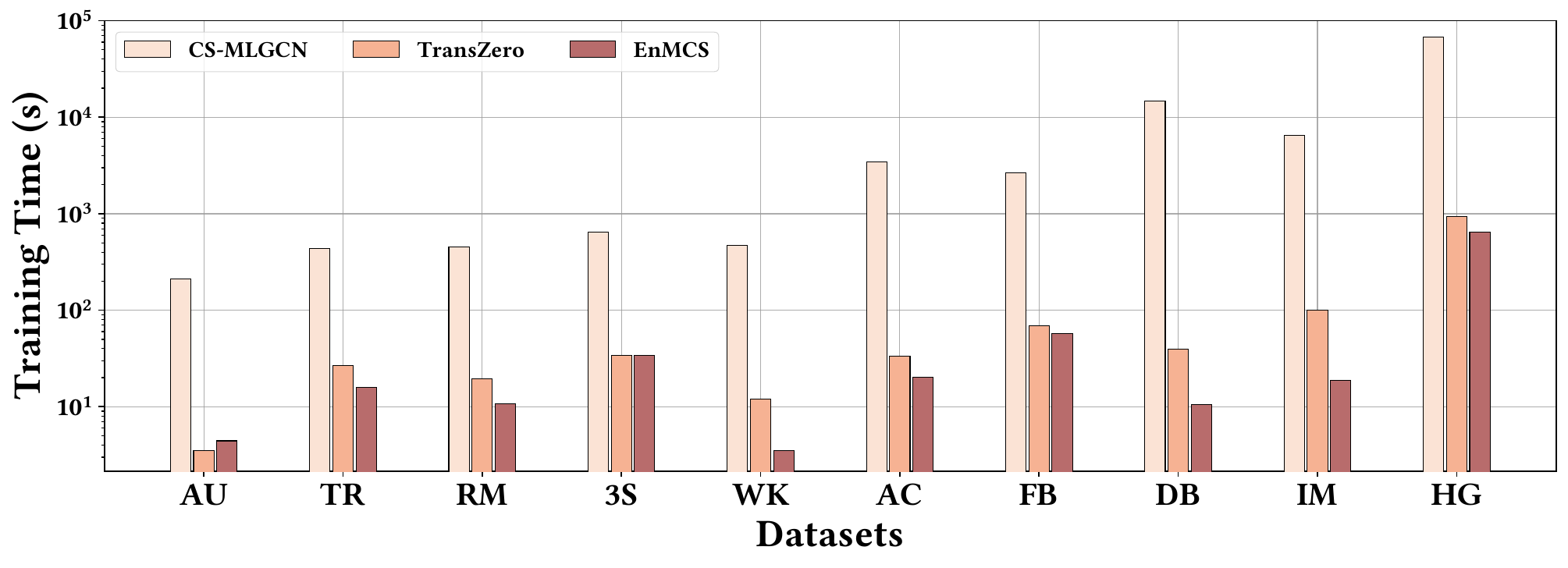}
    }
    \subfigure[Efficiency results of the search phase]{ 
        \includegraphics[width=0.485\textwidth]{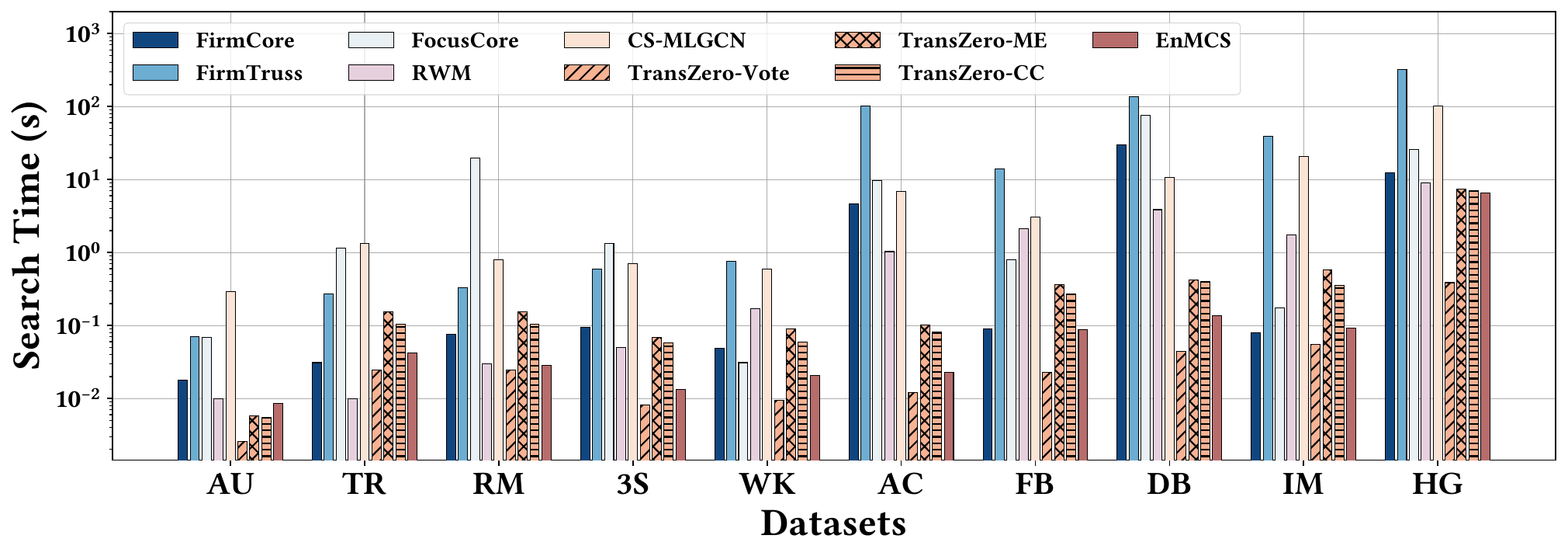}
    }
\vspace{-6mm}
    \caption{{Efficiency results}}
    \label{fig:efficiency_evaluation}
\vspace{-6mm}
\end{figure*}

\begin{figure*}
\subfigbottomskip=-3pt 
\subfigcapskip=-7pt 
    
    \subfigure[F1-score with varying $\lambda$]{ 
        
        \includegraphics[width=0.32\textwidth]{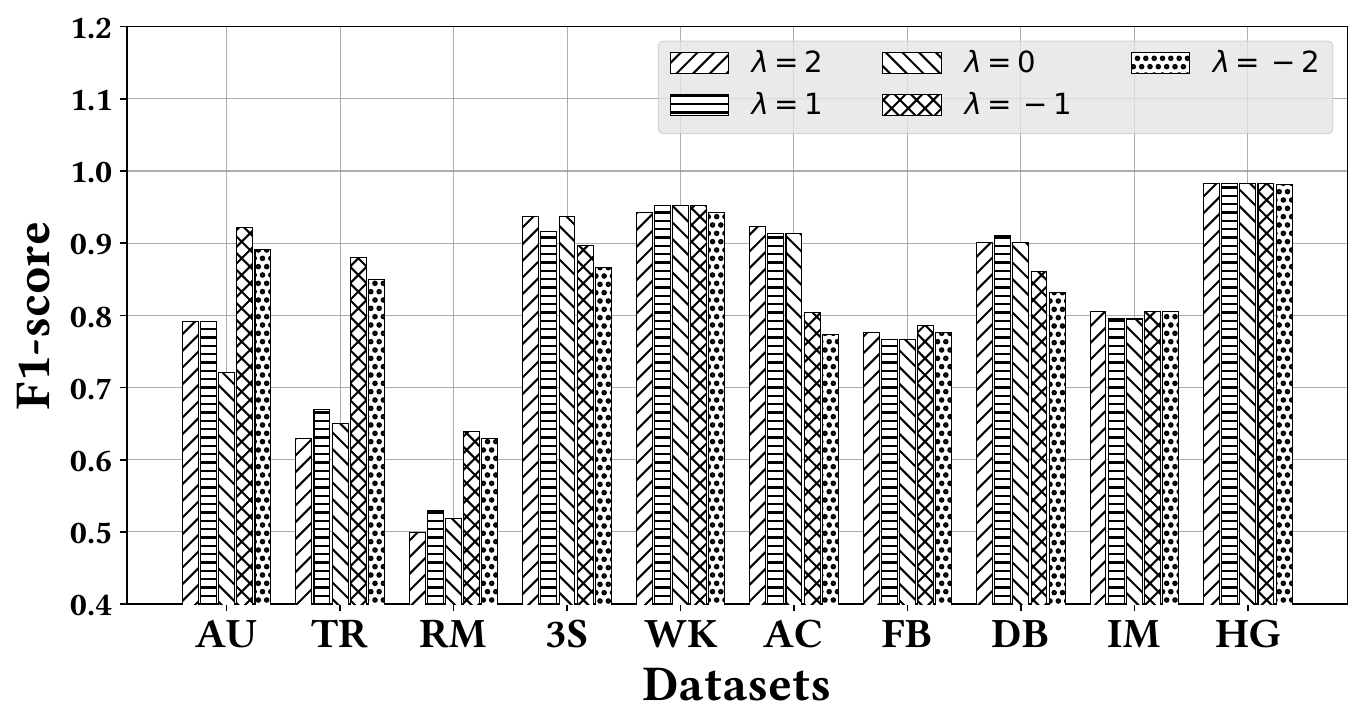}
    \vspace{-3mm}
    }
    \vspace{-1mm}
    \subfigure[F1-score with varying $\tau$]{ 
        
        \includegraphics[width=0.32\textwidth]{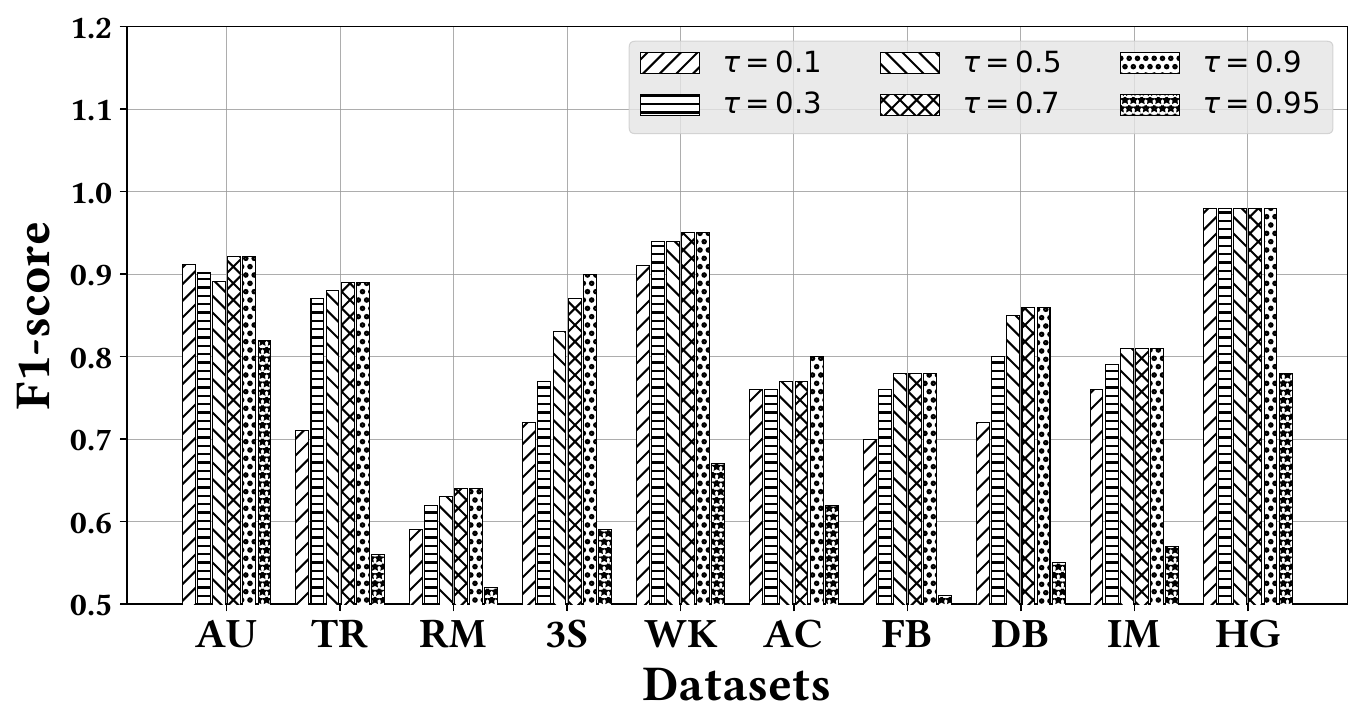}
    \vspace{-3mm}
    }
    \subfigure[F1-score with varying similarity definitions]{
        \includegraphics[width=0.32\textwidth]{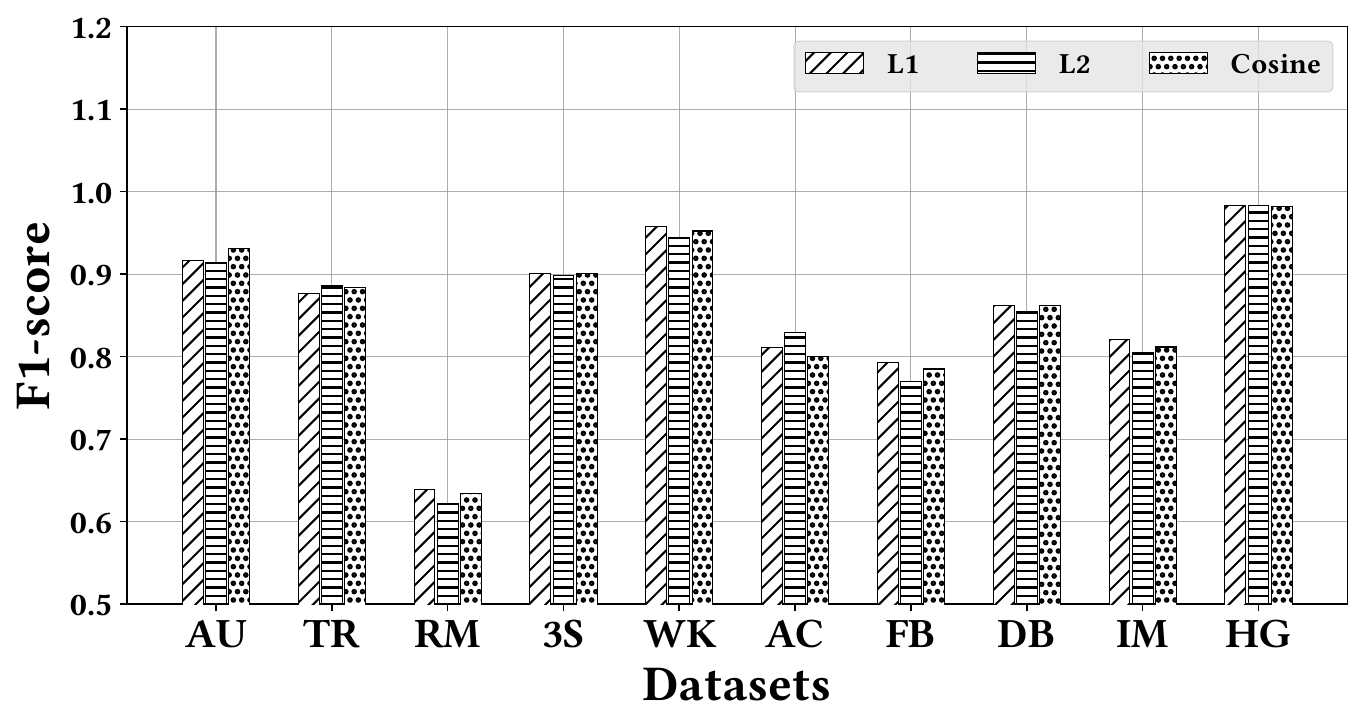} 
    \vspace{-3mm}
    }
    \subfigure[F1-score with varying epochs]{
        \vspace{-5mm}
        \includegraphics[width=0.32\textwidth]{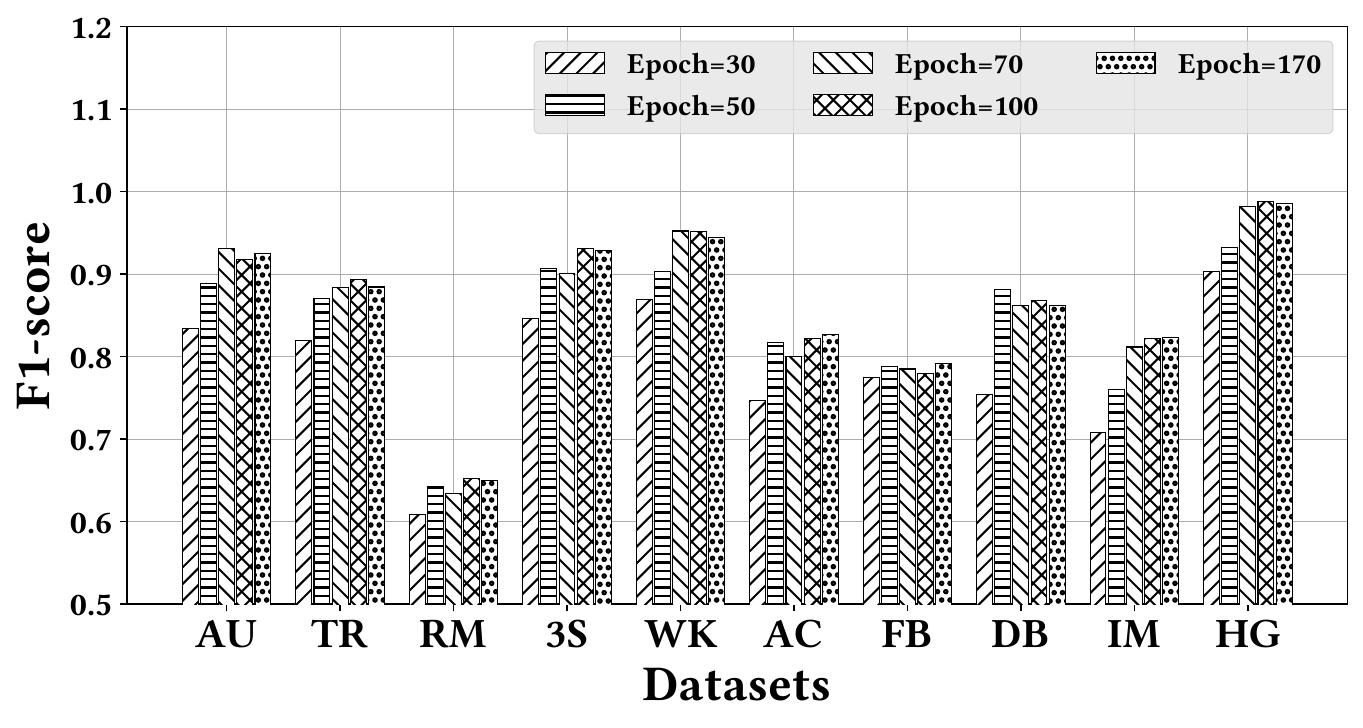} 
    \vspace{-3mm}
    }
    \subfigure[F1-score with varying hops]{
        \vspace{-5mm}
        \includegraphics[width=0.32\textwidth]{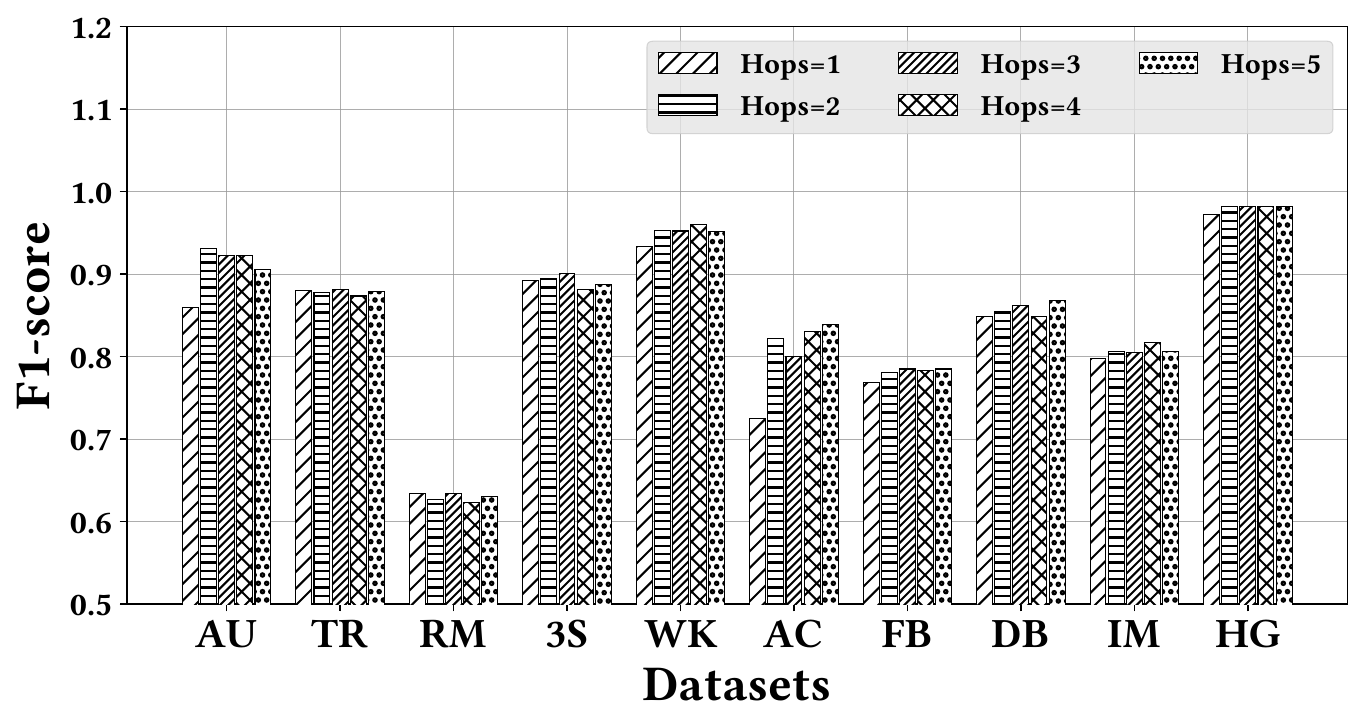} 
    \vspace{-3mm}
    }
    \subfigure[$\alpha$ and $\beta$ analysis]{ 
        \vspace{-5mm}
        \includegraphics[width=0.32\textwidth]{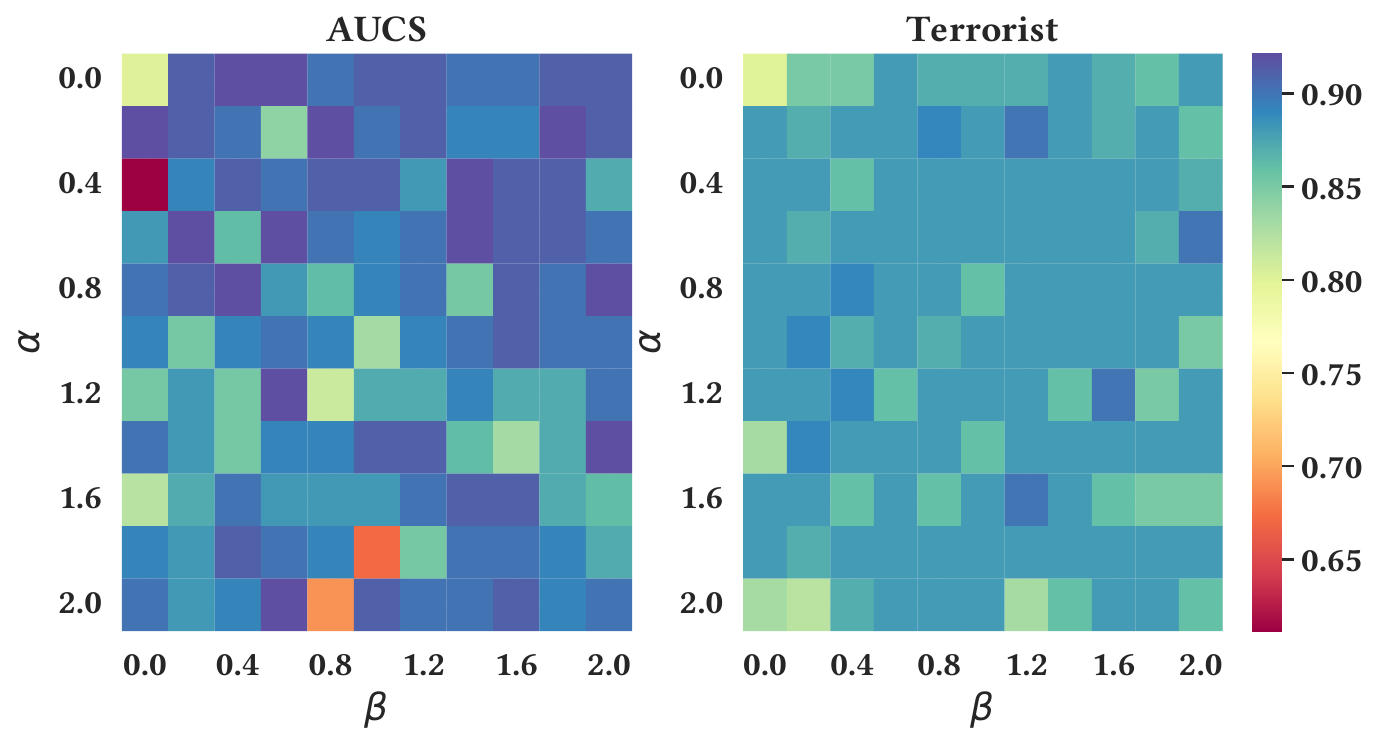}
    \vspace{-3mm}
    }
\vspace{-3mm}
    \caption{Results of hyper-parameter analysis }
    \label{fig:hyperparameter_analysis}
\vspace{-6mm}
\end{figure*}

\noindent\textbf{Implementation Details:} 
We set the hidden dimension to 512. The balancing parameter $\lambda$ is set to -1, and the granularity control parameter $\tau$ is set to 0.9 by default. \framework~ is trained for 70 epochs with early stopping. We set $\alpha$ and $\beta$, which balance the inter-layer loss and intra-layer loss, to 0.8 and 0.4, respectively. The number of neighbor hops is fixed at 3. FFNs consist of 3 layers with 1 hidden layer. Tolerance $\epsilon$ is set as 0.00001. The learning rate is initialized at 1e-4, with a peak value of 0.01. The Adam optimizer is employed with a weight decay of 1e-4. Hyperparameters for baseline models follow the settings provided in their original papers. All experiments were conducted on a server equipped with an Intel Xeon(R) Silver 4314 CPU and Nvidia A5000 GPUs.

\vspace{-3mm}
\subsection{Overall Effectiveness}
\vspace{-1mm}
\noindent\textbf{Exp-1: F1-score results.} 
In this experiment, we evaluate the effectiveness of methods, with the F1-score results summarized in Table~\ref{tab:exp1_f1score}. The baselines in inductive, transductive and hybrid settings are denoted as IN, TD and HB, respectively. Rule-based methods, RWM, TransZero and \framework~ show consistent results across all settings as they are label-free. Hence, we only present one set of results to avoid redundancy. 
The table shows that FirmTruss demonstrates outstanding performance among the rule-based methods. Additionally, CS-MLGCN performs better under the transductive setting than inductive, suggesting its limited generalization capacity for predicting unseen communities. 
The multilayer graph in the RM dataset is highly dense and contains multiple disjoint components, which can be easily captured by rule-based methods but make it challenging for \framework~to aggregate global information effectively, resulting in an unsatisfactory performance of \framework~ on RM. However, overall, \framework~significantly outperforms the baselines, achieving average F1-score improvements of 8.78\%, 19.41\% and 39.24\% compared to TransZero-CC, CS-MLGCN (TD) and FirmTruss, respectively.

\vspace{-3mm}
\subsection{Efficiency Evaluation}


\noindent\textbf{Exp-2: Efficiency results.}  
In this experiment, we evaluate the efficiency of baselines, focusing on both training and search time. The results are presented in Figure~\ref{fig:efficiency_evaluation}. Training time is reported for CS-MLGCN, TransZero, and \framework, as these methods require model training, while search time is measured for all baselines. As illustrated in the figure, \framework~delivers outstanding performance in both phases. In the training phase, \framework~is 233$\times$ faster than CS-MLGCN and 2.2$\times$ faster than TransZero. In terms of search time, although EMerge requires additional time, \framework~still delivers competitive performance compared to TransZero and CS-MLGCN. It achieves an average speedup of 118$\times$ compared to CS-MLGCN, with a maximum speedup of 452$\times$ on the ACM dataset.

\vspace{-3mm}
\subsection{Hyper-parameter Analysis}
\noindent\textbf{Exp-3: Varying $\lambda$.} In this experiment, we assess the impact of $\lambda$ which balances the layer-shared and layer-specific community scores. We test $\lambda$ values of 2, 1, 0, -1, and -2, with the results shown in Figure~\ref{fig:hyperparameter_analysis}(a). As illustrated, in datasets like ACM, performance decreases as $\lambda$ decreases, whereas in datasets like AUCS and Terrorist, performance improves with lower $\lambda$ values. Overall, setting $\lambda$ to -1 achieves strong results, with an average F1-score improvement of 3.64\% over using only the layer-shared score (i.e., $\lambda$ = 0), demonstrating the effectiveness of the layer-shared representations and layer-specific representations for multilayer CS.

\noindent\textbf{Exp-4: Varying $\tau$.} In this experiment, we evaluate the effect of $\tau$, which controls the granularity of the subgraph. We test $\tau$ with values of 0.1, 0.3, 0.5, 0.7, 0.9 and 0.95, and the results are displayed in Figure~\ref{fig:hyperparameter_analysis}(b). As shown, most datasets exhibit better performance as $\tau$ increases when it is below 0.9. However, beyond 0.9, performance decreases. Hence, we set $\tau$ to 0.9.

\noindent\textbf{Exp-5: Varying similarity definitions.} In this experiment, we evaluate \framework~ using different similarity metrics. We replace Line 4 and Line 5 in Algorithm~\ref{algo:score_computation} with L1 and L2 similarities~\cite{DBLP:journals/pacmmod/WangZW0024}, and the results are presented in Figure~\ref{fig:hyperparameter_analysis}(c). The results indicate that the performance remains consistent, demonstrating the robustness of \framework~ to varying similarity definitions.

\noindent\textbf{Exp-6: Varying epoch numbers.} In this experiment, we assess the impact of the number of epochs to train \framework. We test epochs with values of 30, 50, 70, 100 and 130, and the results are shown in Figure~\ref{fig:hyperparameter_analysis}(d). As illustrated, most datasets show improved performance as the number of epochs increases, particularly below 70. However, beyond 70, performance stabilizes. Therefore, we set the number of training epochs to 70.

\noindent\textbf{Exp-7: Varying hops.} In this experiment, we evaluate the impact of the number of hops, which provides contextual information in Eqn~\ref{eqn:hops}. We test hop values of 1, 2, 3, 4 and 5, with the results shown in Figure~\ref{fig:hyperparameter_analysis}(e). As illustrated, most datasets exhibit improved performance as the number of hops increases when the hop number is less than 3. Beyond 3 hops, performance stabilizes, indicating that 3 hops are sufficient for effective model training.

\noindent\textbf{Exp-8: Varying $\alpha$ and $\beta$.}
In this experiment, we assess the impact of $\alpha$ and $\beta$, which are used to balance the loss functions in Eqn~\ref{equ:loss_func}. We test values of both $\alpha$ and $\beta$ ranging from 0 to 2, with an interval of 0.2. The results for the AUCS and Terrorist datasets are presented in Figure~\ref{fig:hyperparameter_analysis}(f). As shown, non-zero values of $\alpha$ and $\beta$ generally yield better performance compared to zero. Furthermore, setting $\alpha$ to 0.8 and $\beta$ to 0.4 typically shows a good performance.

\vspace{-4mm}
\subsection{Ablation Study and Case Study}
\vspace{-1mm}
\noindent\textbf{Exp-9: Ablation study.}
In this experiment, we evaluate the effectiveness of components designed in \framework, including $\mathcal{L}_p$, $\mathcal{L}_{inter}$, $\mathcal{L}_{intra}$, EMerge and graph diffusion encoder. The first four components are designed to improve accuracy, and the graph diffusion encoder is designed to enhance efficiency while maintaining high accuracy. The results are summarized in Figure~\ref{fig:ablation_evaluation}. To evaluate EMerge, we replace it with the voting strategy. The results show that all five components contribute to the high performance of \framework, with $\mathcal{L}_p$, $\mathcal{L}_{inter}$, $\mathcal{L}_{intra}$ and EMerge improving the average F1-score by 1.12\%, 2.13\%, 2.78\% and 4.65\%, respectively. To assess the efficiency of graph diffusion encoder, we replace it with the multilayer GCN used in CS-MLGCN, and the results show that graph diffusion encoder achieves an average speedup of 144$\times$. These results collectively highlight the effectiveness of the proposed components. 

\begin{figure*}
\subfigcapskip=-7pt 
    
    \subfigure[F1-score results of the training phase]{ 
        
        \includegraphics[width=0.485\textwidth]{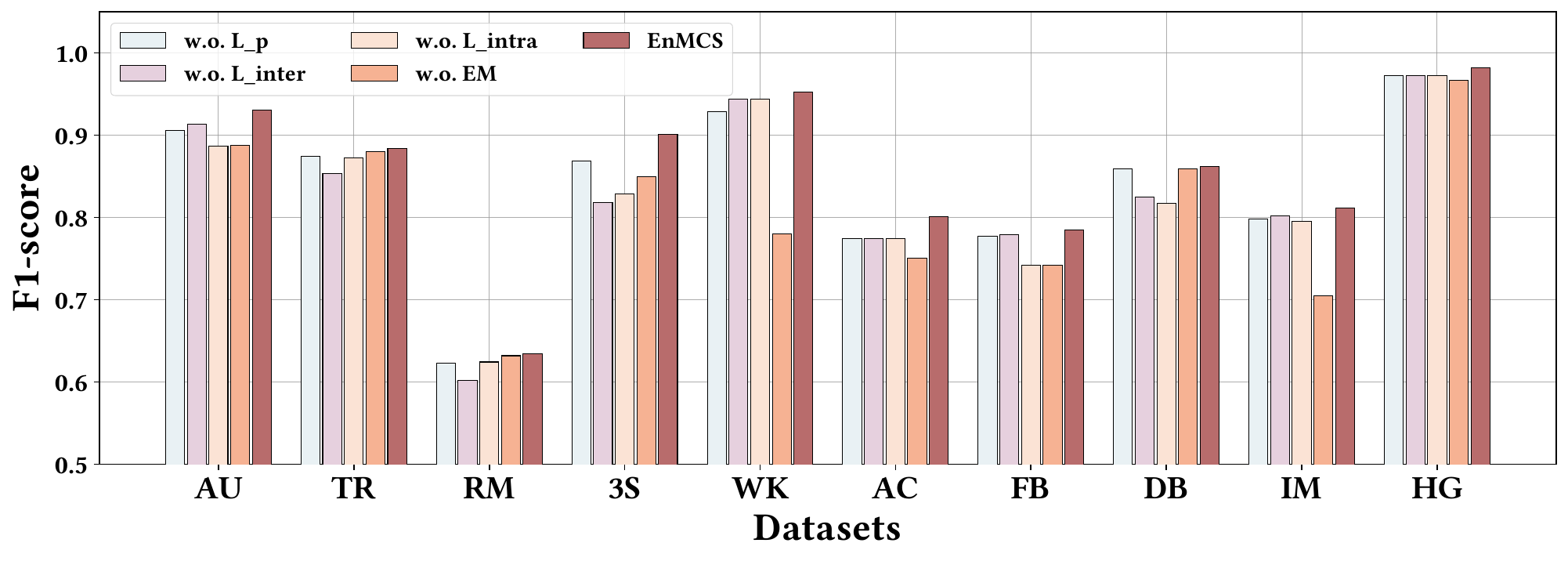}
    }
    \subfigure[Efficiency results of the search phase]{ 
        \includegraphics[width=0.485\textwidth]{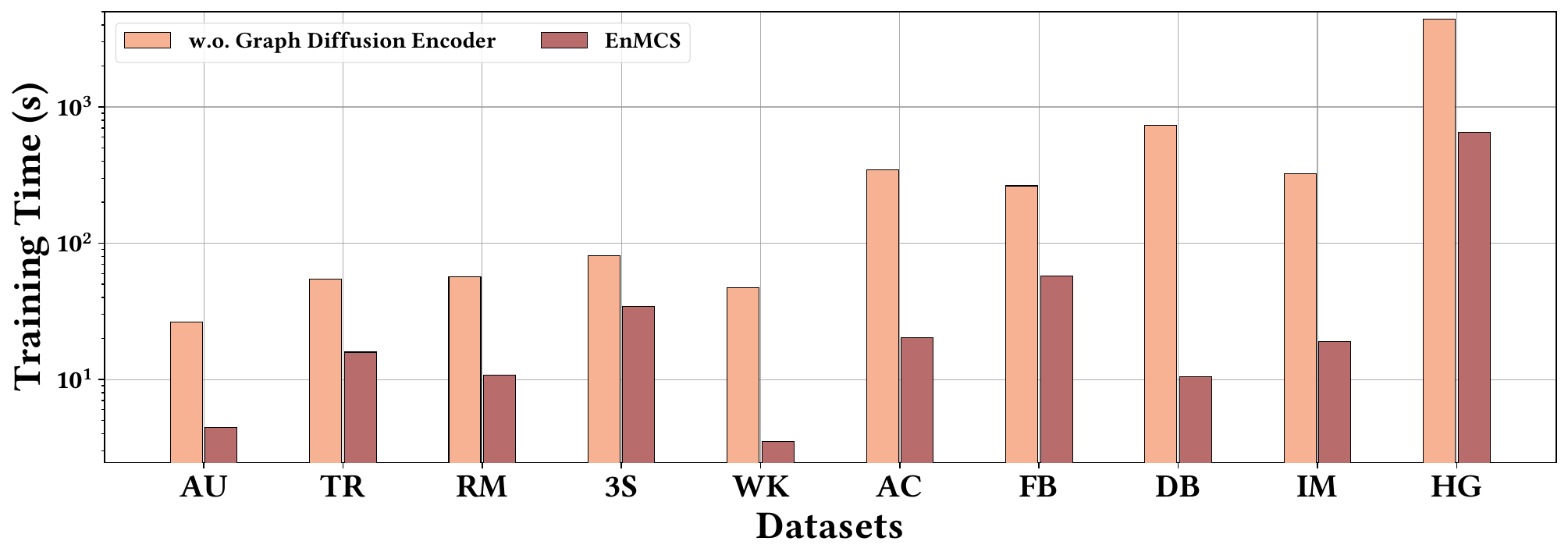}
    }
\vspace{-5mm}
    \caption{{Ablation study}}
    \label{fig:ablation_evaluation}
\vspace{-4mm}
\end{figure*}
\begin{figure*}
\vspace{-2mm}
    \subfigcapskip=-4pt 
    \subfigure[\textbf{Gound-truth}]{ 
        \centering
        \includegraphics[width=0.20\textwidth]{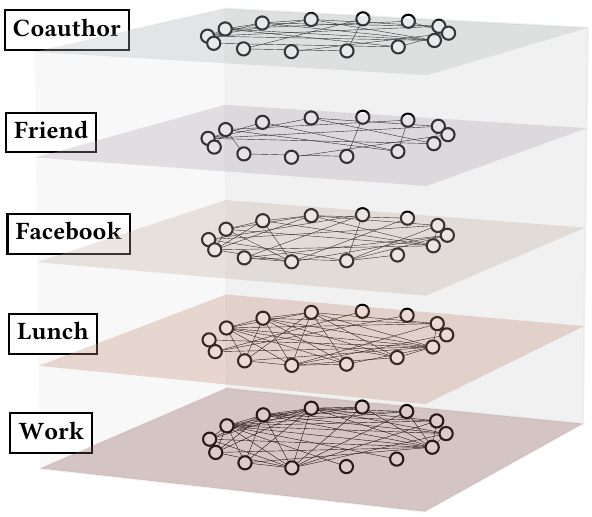}
        \captionsetup{skip=10pt}
      }
    \hfill
    \subfigure[\textbf{\framework}]{
        \centering
        \includegraphics[width=0.20\textwidth]{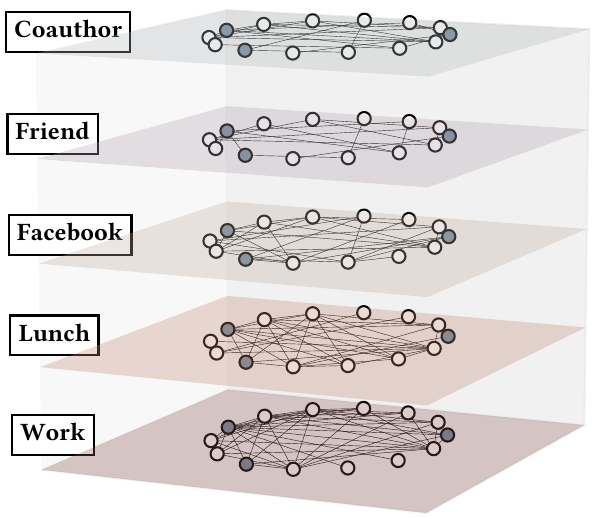} 
    }
    \hfill
    \subfigure[\textbf{TransZero-CC}]{
        \centering
        \includegraphics[width=0.20\textwidth]{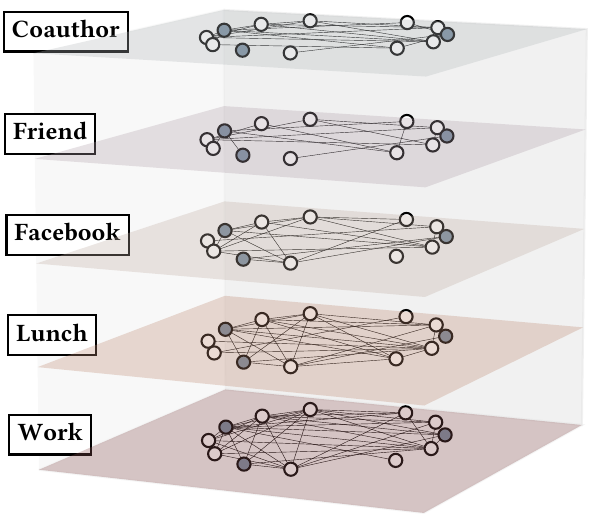} 
    }
    \hfill
    \subfigure[\textbf{CS-MLGCN}]{
        \centering
        \includegraphics[width=0.20\textwidth]{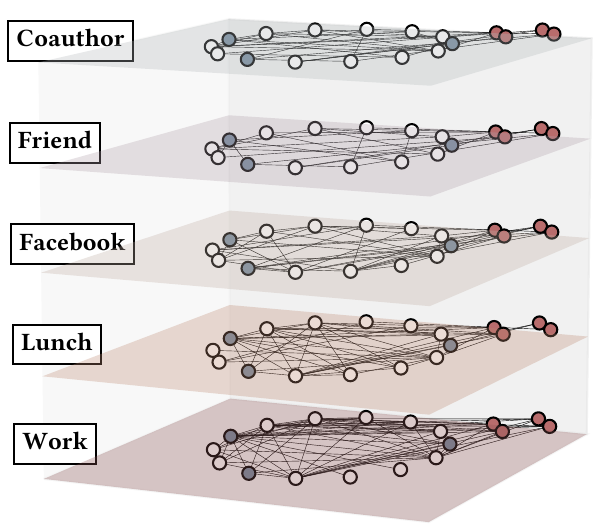} 
    }
    \vspace{-6mm}
    \caption{Case study}
    \label{fig:case_study}
\vspace{-5mm}
\end{figure*}
{\color{black}
\noindent\textbf{Exp-10: Case study}. We conducted a case study using the real-world AUCS multilayer graph~\cite{kim2016differential}, which consists of five layers (dimensions) of relationships: Coauthor, Friend, Facebook, Lunch and Work. Both the ground-truth community and the results obtained from learning-based methods are depicted in Figure~\ref{fig:case_study}, with the query node highlighted in blue. The illustration reveals that the TransZero-CC method fails to include some promising nodes, while CS-MLGCN includes several irrelevant nodes. On the other hand, the community identified by our proposed \framework~accurately matches the ground-truth community.

}

\vspace{-3mm}
\section{Related Work}
\vspace{-1mm}
\label{sec:Relatedwork}
In this section, we review the relevant literature, focusing on two closely related areas: multilayer CS and single-layer CS.

\noindent\textbf{Multilayer Community Search}.
A suite of multilayer CS methods evolved from single-layer models, has been introduced to more effectively address the intricacies of inter-layer correlations. FirmCore~\cite{hashemi2022firmcore} and FocusCore~\cite{wang2024focuscore} extend the $k$-core concept from single-layer to multilayer graphs. FirmCore defines a multilayer community as at least $\lambda$ layers that satisfy the core number constraints. FocusCore further refines this definition by requiring that the $\lambda$ layers specifically include a set of predefined layers (focus layers), thereby increasing the specificity of the search process. Similarly, FirmTruss~\cite{akbas2017truss} extends the $k$-truss structure to multilayer settings, requiring that at least $\lambda$ layers meet the truss number constraints to ensure structural consistency across layers.
ML-LCD~\cite{interdonato2017local} employs a greedy algorithm to identify local communities within multiplex networks.
Additionally, RWM~\cite{luo2020local} leverages a random walk approach, designing a probability transition mechanism tailored for multilayer community search. 
CS-MLGCN~\cite{behrouz2022cs} designs a learning model that can predict the membership of each node and utilizes labeled communities to train the model.
However, these methods often face challenges such as the lack of layer-specific characteristics and reliance on label information, which limit their effectiveness and applicability in real-world scenarios.

\noindent\textbf{Single-layer Community Search}.
Single-layer CS can be categorized into two groups: traditional CS methods and learning-based techniques. Traditional CS methods primarily focus on searching a cohesively connected subgraph~\cite{zhao2022efficient} within a data graph, incorporating specific query nodes and adhering to predefined constraints. These methods traditionally employ cohesive subgraph models such as \textit{k}-core~\cite{cui2014local, wang2018efficient}, \textit{k}-truss~\cite{huang2014querying, akbas2017truss} and \textit{k}-edge connected component (\textit{k}-ECC)~\cite{chang2013efficiently, hu2016querying} to model communities. However, a significant drawback of these approaches is their \textit{structural inflexibility} as fixed rules are hard to capture real-world communities, which limits their applications.
Recent advances have spurred interest in learning-based CS methods, which leverage deep learning techniques to enhance the effectiveness and flexibility of CS~\cite{sima2024deep}. For instance, ICS-GNN~\cite{gao2021ics} introduces a lightweight, interactive CS model utilizing a graph neural network. Additionally, QD-GNN and AQD-GNN presented in~\cite{jiang2022query}, offer solutions for CS and attributed CS within a supervised framework. COCLEP, proposed in~\cite{li2023coclep}, facilitates CS in a semi-supervised context, requiring only a small set of labeled nodes within the community as opposed to complete labeling of all nodes in the ground-truth community. Meta-learning~\cite{vilalta2002perspective} is introduced for inductive CS in ~\cite{fang2024inductive, fang2023community}.
Moreover, TransZero~\cite{wang2024efficient} is an innovative unsupervised CS approach that works without the need of labels by integrating the offline phase and online phase. 
Despite these advancements, these methods are primarily designed for single-layer graphs and struggle to capture the inter-layer relationships which are critical for effective multilayer CS.

\vspace{-4mm}

\section{Conclusion}
\vspace{-1mm}
\label{sec:Conclusion}

In this paper, we study the problem of unsupervised multilayer community search and propose a novel ensemble-based framework, \framework, which integrates \singlesearchname~and \combinename. \singlesearchname~is designed to search communities in each layer while considering both inter- and intra-layer information. It leverages a graph-diffusion model to learn common and private representations, guided by three label-free losses: proximity loss, matching loss, and correlation loss. Communities are then identified in each layer based on similarity. \combinename~merges the communities from each layer into a final consensus community without requiring labels. Using an EM-based algorithm, it simultaneously estimates both true node assignments and layer error rates, optimizing these through iterative maximum likelihood estimation. This "search and merge" approach offers flexibility in updating components and provides a new perspective for generalizing single-layer CS methods to the multilayer setting. 
Extensive experiments on 10 real-world multilayer graphs demonstrate the superior accuracy and competitive efficiency of \framework~compared to previous multilayer CS methods.

\balance
{
\bibliographystyle{ACM-Reference-Format}
\bibliography{sample}

@article{wasserman1994social,
  title={Social network analysis: Methods and applications},
  author={Wasserman, Stanley and Faust, Katherine},
  year={1994},
  publisher={Cambridge university press}
}

@book{knoke2008social,
  title={Social network analysis},
  author={Knoke, David and Yang, Song},
  number={154},
  year={2008},
  publisher={Sage}
}

@inproceedings{topchy2004mixture,
  title={A mixture model for clustering ensembles},
  author={Topchy, Alexander and Jain, Anil K and Punch, William},
  booktitle={Proceedings of the 2004 SIAM international conference on data mining},
  pages={379--390},
  year={2004},
  organization={SIAM}
}

@article{sima2024deep,
  title={Deep Overlapping Community Search via Subspace Embedding},
  author={Sima, Qing and Yu, Jianke and Wang, Xiaoyang and Zhang, Wenjie and Zhang, Ying and Lin, Xuemin},
  journal={arXiv preprint arXiv:2404.14692},
  year={2024}
}

@inproceedings{wang2022collaborative,
  title={Collaborative knowledge distillation for heterogeneous information network embedding},
  author={Wang, Can and Zhou, Sheng and Yu, Kang and Chen, Defang and Li, Bolang and Feng, Yan and Chen, Chun},
  booktitle={Proceedings of the ACM Web Conference 2022},
  pages={1631--1639},
  year={2022}
}

@inproceedings{zhu2022structure,
  title={Structure-enhanced heterogeneous graph contrastive learning},
  author={Zhu, Yanqiao and Xu, Yichen and Cui, Hejie and Yang, Carl and Liu, Qiang and Wu, Shu},
  booktitle={Proceedings of the 2022 SIAM International Conference on Data Mining (SDM)},
  pages={82--90},
  year={2022},
  organization={SIAM}
}

@article{eboli2004systemic,
  title={Systemic risk in financial networks: a graph theoretic approach},
  author={Eboli, Mario and others},
  journal={Universita di Chieti Pescara},
  year={2004},
  publisher={Citeseer}
}

@article{heimo2009maximal,
  title={Maximal spanning trees, asset graphs and random matrix denoising in the analysis of dynamics of financial networks},
  author={Heimo, Tapio and Kaski, Kimmo and Saram{\"a}ki, Jari},
  journal={Physica A: Statistical Mechanics and its Applications},
  volume={388},
  number={2-3},
  pages={145--156},
  year={2009},
  publisher={Elsevier}
}

@article{wang2024neural,
  title={Neural Attributed Community Search at Billion Scale},
  author={Wang, Jianwei and Wang, Kai and Lin, Xuemin and Zhang, Wenjie and Zhang, Ying},
  journal={Proceedings of the ACM on Management of Data},
  volume={1},
  number={4},
  pages={1--25},
  year={2024},
  publisher={ACM New York, NY, USA}
}

@article{wang2024efficient,
  author       = {Jianwei Wang and
                  Kai Wang and
                  Xuemin Lin and
                  Wenjie Zhang and
                  Ying Zhang},
  title        = {Efficient Unsupervised Community Search with Pre-trained Graph Transformer},
  journal      = {Proc. {VLDB} Endow.},
  volume       = {17},
  number       = {9},
  pages        = {2227--2240},
  year         = {2024},
  url          = {https://www.vldb.org/pvldb/vol17/p2227-wang.pdf},
  doi          = {10.14778/3665844.3665853},
  bibsource    = {dblp computer science bibliography, https://dblp.org}
}

@article{he2023scaling,
  title={Scaling up k-clique densest subgraph detection},
  author={He, Yizhang and Wang, Kai and Zhang, Wenjie and Lin, Xuemin and Zhang, Ying},
  journal={Proceedings of the ACM on Management of Data},
  volume={1},
  number={1},
  pages={1--26},
  year={2023},
  publisher={ACM New York, NY, USA}
}

@article{DBLP:journals/pacmmod/WangZW0024,
  author       = {Jianwei Wang and
                  Ying Zhang and
                  Kai Wang and
                  Xuemin Lin and
                  Wenjie Zhang},
  title        = {Missing Data Imputation with Uncertainty-Driven Network},
  journal      = {Proc. {ACM} Manag. Data},
  volume       = {2},
  number       = {3},
  pages        = {117},
  year         = {2024},
  url          = {https://doi.org/10.1145/3654920},
  doi          = {10.1145/3654920},
  bibsource    = {dblp computer science bibliography, https://dblp.org}
}

@article{behrouz2022firmtruss,
  title={FirmTruss Community Search in Multilayer Networks},
  author={Behrouz, Ali and Hashemi, Farnoosh and Lakshmanan, Laks VS},
  journal={Proceedings of the VLDB Endowment},
  volume={16},
  number={3},
  pages={505--518},
  year={2022},
  publisher={VLDB Endowment}
}

@inproceedings{li2023coclep,
  title={COCLEP: Contrastive Learning-based Semi-Supervised Community Search},
  author={Li, Ling and Luo, Siqiang and Zhao, Yuhai and Shan, Caihua and Wang, Zhengkui and Qin, Lu},
  booktitle={2023 IEEE 39th International Conference on Data Engineering (ICDE)},
  pages={2483--2495},
  year={2023},
  organization={IEEE}
}

@article{jiang2022query,
  title={Query driven-graph neural networks for community search: from non-attributed, attributed, to interactive attributed},
  author={Jiang, Yuli and Rong, Yu and Cheng, Hong and Huang, Xin and Zhao, Kangfei and Huang, Junzhou},
  journal={Proceedings of the VLDB Endowment},
  volume={15},
  number={6},
  pages={1243--1255},
  year={2022},
  publisher={VLDB Endowment}
}

@article{interdonato2017local,
  title={Local community detection in multilayer networks},
  author={Interdonato, Roberto and Tagarelli, Andrea and Ienco, Dino and Sallaberry, Arnaud and Poncelet, Pascal},
  journal={Data Mining and Knowledge Discovery},
  volume={31},
  pages={1444--1479},
  year={2017},
  publisher={Springer}
}

@inproceedings{andersen2006local,
  title={Local graph partitioning using pagerank vectors},
  author={Andersen, Reid and Chung, Fan and Lang, Kevin},
  booktitle={2006 47th Annual IEEE Symposium on Foundations of Computer Science (FOCS'06)},
  pages={475--486},
  year={2006},
  organization={IEEE}
}

@article{kim2015community,
  title={Community detection in multi-layer graphs: A survey},
  author={Kim, Jungeun and Lee, Jae-Gil},
  journal={ACM SIGMOD Record},
  volume={44},
  number={3},
  pages={37--48},
  year={2015},
  publisher={ACM New York, NY, USA}
}

@book{fujishige2005submodular,
  title={Submodular functions and optimization},
  author={Fujishige, Satoru},
  year={2005},
  publisher={Elsevier}
}

@article{hao2023ensemble,
  title={Ensemble clustering with attentional representation},
  author={Hao, Zhezheng and Lu, Zhoumin and Li, Guoxu and Nie, Feiping and Wang, Rong and Li, Xuelong},
  journal={IEEE Transactions on Knowledge and Data Engineering},
  year={2023},
  publisher={IEEE}
}

@inproceedings{schroff2015facenet,
  title={Facenet: A unified embedding for face recognition and clustering},
  author={Schroff, Florian and Kalenichenko, Dmitry and Philbin, James},
  booktitle={Proceedings of the IEEE conference on computer vision and pattern recognition},
  pages={815--823},
  year={2015}
}

@article{wang2024disensemi,
  title={Disensemi: Semi-supervised graph classification via disentangled representation learning},
  author={Wang, Yifan and Luo, Xiao and Chen, Chong and Hua, Xian-Sheng and Zhang, Ming and Ju, Wei},
  journal={IEEE Transactions on Neural Networks and Learning Systems},
  year={2024},
  publisher={IEEE}
}

@inproceedings{DBLP:conf/icde/WangYC00L23,
  author       = {Yiqi Wang and
                  Long Yuan and
                  Zi Chen and
                  Wenjie Zhang and
                  Xuemin Lin and
                  Qing Liu},
  title        = {Towards Efficient Shortest Path Counting on Billion-Scale Graphs},
  booktitle    = {39th {IEEE} International Conference on Data Engineering, {ICDE} 2023,
                  Anaheim, CA, USA, April 3-7, 2023},
  pages        = {2579--2592},
  publisher    = {{IEEE}},
  year         = {2023},
  url          = {https://doi.org/10.1109/ICDE55515.2023.00198},
  doi          = {10.1109/ICDE55515.2023.00198},
  bibsource    = {dblp computer science bibliography, https://dblp.org}
}

@article{DBLP:journals/corr/abs-2408-05432,
  author       = {Yiqi Wang and
                  Long Yuan and
                  Wenjie Zhang and
                  Xuemin Lin and
                  Zi Chen and
                  Qing Liu},
  title        = {Simpler is More: Efficient Top-K Nearest Neighbors Search on Large
                  Road Networks},
  journal      = {CoRR},
  volume       = {abs/2408.05432},
  year         = {2024},
  url          = {https://doi.org/10.48550/arXiv.2408.05432},
  doi          = {10.48550/ARXIV.2408.05432},
  eprinttype    = {arXiv},
  eprint       = {2408.05432},
  bibsource    = {dblp computer science bibliography, https://dblp.org}
}

@article{xu2020stacked,
  title={Stacked autoencoder-based community detection method via an ensemble clustering framework},
  author={Xu, Rongbin and Che, Yan and Wang, Xinmei and Hu, Jianxiong and Xie, Ying},
  journal={Information sciences},
  volume={526},
  pages={151--165},
  year={2020},
  publisher={Elsevier}
}

@article{wang2022stable,
  title={Stable label-specific features generation for multi-label learning via mixture-based clustering ensemble},
  author={Wang, Yi-Bo and Hang, Jun-Yi and Zhang, Min-Ling},
  journal={IEEE/CAA Journal of Automatica Sinica},
  volume={9},
  number={7},
  pages={1248--1261},
  year={2022},
  publisher={IEEE}
}

@article{tagarelli2017ensemble,
  title={Ensemble-based community detection in multilayer networks},
  author={Tagarelli, Andrea and Amelio, Alessia and Gullo, Francesco},
  journal={Data Mining and Knowledge Discovery},
  volume={31},
  pages={1506--1543},
  year={2017},
  publisher={Springer}
}

@inproceedings{behrouz2022cs,
  title={Cs-mlgcn: Multiplex graph convolutional networks for community search in multiplex networks},
  author={Behrouz, Ali and Hashemi, Farnoosh},
  booktitle={Proceedings of the 31st ACM International Conference on Information \& Knowledge Management},
  pages={3828--3832},
  year={2022}
}

@inproceedings{huang2014querying,
  title={Querying k-truss community in large and dynamic graphs},
  author={Huang, Xin and Cheng, Hong and Qin, Lu and Tian, Wentao and Yu, Jeffrey Xu},
  booktitle={Proceedings of the 2014 ACM SIGMOD international conference on Management of data},
  pages={1311--1322},
  year={2014}
}

@article{akbas2017truss,
  title={Truss-based community search: a truss-equivalence based indexing approach},
  author={Akbas, Esra and Zhao, Peixiang},
  journal={Proceedings of the VLDB Endowment},
  volume={10},
  number={11},
  pages={1298--1309},
  year={2017},
  publisher={VLDB Endowment}
}

@inproceedings{luo2020local,
  title={Local community detection in multiple networks},
  author={Luo, Dongsheng and Bian, Yuchen and Yan, Yaowei and Liu, Xiao and Huan, Jun and Zhang, Xiang},
  booktitle={Proceedings of the 26th ACM SIGKDD international conference on knowledge discovery \& data mining},
  pages={266--274},
  year={2020}
}

@inproceedings{wang2024focuscore,
  title={FocusCore Decomposition of Multilayer Graphs},
  author={Wang, Run-An and Liu, Dandan and Zou, Zhaonian},
  booktitle={2024 IEEE 40th International Conference on Data Engineering (ICDE)},
  pages={2792--2804},
  year={2024},
  organization={IEEE}
}

@inproceedings{cui2014local,
  title={Local search of communities in large graphs},
  author={Cui, Wanyun and Xiao, Yanghua and Wang, Haixun and Wang, Wei},
  booktitle={Proceedings of the 2014 ACM SIGMOD international conference on Management of data},
  pages={991--1002},
  year={2014}
}

@inproceedings{hashemi2022firmcore,
  title={Firmcore decomposition of multilayer networks},
  author={Hashemi, Farnoosh and Behrouz, Ali and Lakshmanan, Laks VS},
  booktitle={Proceedings of the ACM Web Conference 2022},
  pages={1589--1600},
  year={2022}
}

@inproceedings{hu2016querying,
  title={Querying minimal steiner maximum-connected subgraphs in large graphs},
  author={Hu, Jiafeng and Wu, Xiaowei and Cheng, Reynold and Luo, Siqiang and Fang, Yixiang},
  booktitle={Proceedings of the 25th ACM International on Conference on Information and Knowledge Management},
  pages={1241--1250},
  year={2016}
}

@inproceedings{gretton2005kernel,
  title={Kernel constrained covariance for dependence measurement},
  author={Gretton, Arthur and Smola, Alexander and Bousquet, Olivier and Herbrich, Ralf and Belitski, Andrei and Augath, Mark and Murayama, Yusuke and Pauls, Jon and Sch{\"o}lkopf, Bernhard and Logothetis, Nikos},
  booktitle={International Workshop on Artificial Intelligence and Statistics},
  pages={112--119},
  year={2005},
  organization={PMLR}
}

@article{joyce2003bayes,
  title={Bayes’ theorem},
  author={Joyce, James},
  year={2003}
}

@inproceedings{mo2023disentangled,
  title={Disentangled multiplex graph representation learning},
  author={Mo, Yujie and Lei, Yajie and Shen, Jialie and Shi, Xiaoshuang and Shen, Heng Tao and Zhu, Xiaofeng},
  booktitle={International Conference on Machine Learning},
  pages={24983--25005},
  year={2023},
  organization={PMLR}
}

@article{gasteiger2019diffusion,
  title={Diffusion improves graph learning},
  author={Gasteiger, Johannes and Wei{\ss}enberger, Stefan and G{\"u}nnemann, Stephan},
  journal={Advances in neural information processing systems},
  volume={32},
  year={2019}
}

@article{kim2016differential,
  title={Differential flattening: A novel framework for community detection in multi-layer graphs},
  author={Kim, Jungeun and Lee, Jae-Gil and Lim, Sungsu},
  journal={ACM Transactions on Intelligent Systems and Technology (TIST)},
  volume={8},
  number={2},
  pages={1--23},
  year={2016},
  publisher={ACM New York, NY, USA}
}

@inproceedings{kondor2002diffusion,
  title={Diffusion kernels on graphs and other discrete structures},
  author={Kondor, Risi Imre and Lafferty, John},
  booktitle={Proceedings of the 19th international conference on machine learning},
  volume={2002},
  pages={315--322},
  year={2002}
}

@book{stirling1730methodus,
  title={Methodus differentialis: sive tractatus de summatione et interpolatione serierum infinitarum},
  author={Stirling, James},
  year={1730}
}

@article{dawid1979maximum,
  title={Maximum likelihood estimation of observer error-rates using the EM algorithm},
  author={Dawid, Alexander Philip and Skene, Allan M},
  journal={Journal of the Royal Statistical Society: Series C (Applied Statistics)},
  volume={28},
  number={1},
  pages={20--28},
  year={1979},
  publisher={Wiley Online Library}
}

@article{gao2021ics,
  title={ICS-GNN: lightweight interactive community search via graph neural network},
  author={Gao, Jun and Chen, Jiazun and Li, Zhao and Zhang, Ji},
  journal={Proceedings of the VLDB Endowment},
  volume={14},
  number={6},
  pages={1006--1018},
  year={2021},
  publisher={VLDB Endowment}
}

@inproceedings{li2024querying,
  title={Querying Historical Cohesive Subgraphs over Temporal Bipartite Graphs},
  author={Li, Shunyang and Wang, Kai and Lin, Xuemin and Zhang, Wenjie and He, Yizhang and Yuan, Long},
  booktitle={2024 IEEE 40th International Conference on Data Engineering (ICDE)},
  pages={2503--2516},
  year={2024},
  organization={IEEE}
}

@inproceedings{DBLP:conf/wise/LiYZZL21,
  author       = {Shunyang Li and
                  Zhengyi Yang and
                  Xianhang Zhang and
                  Wenjie Zhang and
                  Xuemin Lin},
  editor       = {Wenjie Zhang and
                  Lei Zou and
                  Zakaria Maamar and
                  Lu Chen},
  title        = {SQL2Cypher: Automated Data and Query Migration from {RDBMS} to {GDBMS}},
  booktitle    = {Web Information Systems Engineering - {WISE} 2021 - 22nd International
                  Conference on Web Information Systems Engineering, {WISE} 2021, Melbourne,
                  VIC, Australia, October 26-29, 2021, Proceedings, Part {II}},
  series       = {Lecture Notes in Computer Science},
  volume       = {13081},
  pages        = {510--517},
  publisher    = {Springer},
  year         = {2021},
  url          = {https://doi.org/10.1007/978-3-030-91560-5\_39},
  doi          = {10.1007/978-3-030-91560-5\_39},
  bibsource    = {dblp computer science bibliography, https://dblp.org}
}

@article{he2024discovering,
  title={Discovering critical vertices for reinforcement of large-scale bipartite networks},
  author={He, Yizhang and Wang, Kai and Zhang, Wenjie and Lin, Xuemin and Zhang, Ying},
  journal={The VLDB Journal},
  pages={1--26},
  year={2024},
  publisher={Springer}
}

@inproceedings{cheng2011personalized,
  title={Personalized travel recommendation by mining people attributes from community-contributed photos},
  author={Cheng, An-Jung and Chen, Yan-Ying and Huang, Yen-Ta and Hsu, Winston H and Liao, Hong-Yuan Mark},
  booktitle={Proceedings of the 19th ACM international conference on Multimedia},
  pages={83--92},
  year={2011}
}

@inproceedings{sozio2010community,
  title={The community-search problem and how to plan a successful cocktail party},
  author={Sozio, Mauro and Gionis, Aristides},
  booktitle={Proceedings of the 16th ACM SIGKDD international conference on Knowledge discovery and data mining},
  pages={939--948},
  year={2010}
}

@article{fang2020survey,
  title={A survey of community search over big graphs},
  author={Fang, Yixiang and Huang, Xin and Qin, Lu and Zhang, Ying and Zhang, Wenjie and Cheng, Reynold and Lin, Xuemin},
  journal={The VLDB Journal},
  volume={29},
  pages={353--392},
  year={2020},
  publisher={Springer}
}

@inproceedings{zhang2017hidden,
  title={Hidden: hierarchical dense subgraph detection with application to financial fraud detection},
  author={Zhang, Si and Zhou, Dawei and Yildirim, Mehmet Yigit and Alcorn, Scott and He, Jingrui and Davulcu, Hasan and Tong, Hanghang},
  booktitle={Proceedings of the 2017 SIAM international conference on data mining},
  pages={570--578},
  year={2017},
  organization={SIAM}
}

@inproceedings{li2021happens,
  title={What happens behind the scene? Towards fraud community detection in e-commerce from online to offline},
  author={Li, Zhao and Hui, Pengrui and Zhang, Peng and Huang, Jiaming and Wang, Biao and Tian, Ling and Zhang, Ji and Gao, Jianliang and Tang, Xing},
  booktitle={Companion Proceedings of the Web Conference 2021},
  pages={105--113},
  year={2021}
}

@inproceedings{chen2008combinational,
  title={Combinational collaborative filtering for personalized community recommendation},
  author={Chen, Wen-Yen and Zhang, Dong and Chang, Edward Y},
  booktitle={Proceedings of the 14th ACM SIGKDD international conference on Knowledge discovery and data mining},
  pages={115--123},
  year={2008}
}

@article{vilalta2002perspective,
  title={A perspective view and survey of meta-learning},
  author={Vilalta, Ricardo and Drissi, Youssef},
  journal={Artificial intelligence review},
  volume={18},
  pages={77--95},
  year={2002},
  publisher={Springer}
}

@inproceedings{fang2023community,
  title={Community search: a meta-learning approach},
  author={Fang, Shuheng and Zhao, Kangfei and Li, Guanghua and Yu, Jeffrey Xu},
  booktitle={2023 IEEE 39th International Conference on Data Engineering (ICDE)},
  pages={2358--2371},
  year={2023},
  organization={IEEE}
}

@article{fang2024inductive,
  title={Inductive Attributed Community Search: To Learn Communities Across Graphs},
  author={Fang, Shuheng and Zhao, Kangfei and Rong, Yu and Li, Zhixun and Yu, Jeffrey Xu},
  journal={Proceedings of the VLDB Endowment},
  volume={17},
  number={10},
  pages={2576--2589},
  year={2024},
  publisher={VLDB Endowment}
}

@inproceedings{chang2013efficiently,
  title={Efficiently computing k-edge connected components via graph decomposition},
  author={Chang, Lijun and Yu, Jeffrey Xu and Qin, Lu and Lin, Xuemin and Liu, Chengfei and Liang, Weifa},
  booktitle={Proceedings of the 2013 ACM SIGMOD international conference on management of data},
  pages={205--216},
  year={2013}
}

@inproceedings{wang2018efficient,
  title={Efficient computing of radius-bounded k-cores},
  author={Wang, Kai and Cao, Xin and Lin, Xuemin and Zhang, Wenjie and Qin, Lu},
  booktitle={2018 IEEE 34th international conference on data engineering (ICDE)},
  pages={233--244},
  year={2018},
  organization={IEEE}
}

@inproceedings{zhao2022efficient,
  title={Efficient computation of cohesive subgraphs in uncertain bipartite graphs},
  author={Zhao, Gengda and Wang, Kai and Zhang, Wenjie and Lin, Xuemin and Zhang, Ying and He, Yizhang},
  booktitle={2022 IEEE 38th International Conference on Data Engineering (ICDE)},
  pages={2333--2345},
  year={2022},
  organization={IEEE}
}
}

\end{document}